\documentclass[11pt]{article}
\usepackage[T1]{fontenc}
\usepackage[utf8]{inputenc}
\usepackage{lmodern,amsmath,amssymb,amsthm,mathtools,microtype}
\usepackage[letterpaper,margin=1in]{geometry}
\usepackage{booktabs,array}
\usepackage[dvipsnames]{xcolor}
\usepackage{tikz}
\usetikzlibrary{arrows.meta,calc,decorations.pathreplacing,positioning,shapes.geometric,fit,backgrounds}
\usepackage{pgfplots}
\pgfplotsset{compat=1.18}
\usepackage{caption}
\usepackage{enumitem}
\usepackage[hidelinks,hypertexnames=false]{hyperref}
\usepackage[capitalize,nameinlink,noabbrev]{cleveref}
\usepackage{thmtools}
\declaretheorem[name=Theorem,numberwithin=section]{theorem}
\declaretheorem[name=Lemma,sibling=theorem]{lemma}
\declaretheorem[name=Proposition,sibling=theorem]{proposition}
\declaretheorem[name=Corollary,sibling=theorem]{corollary}

\declaretheorem[name=Definition,sibling=theorem]{definition}
\numberwithin{equation}{section}
\crefname{appendix}{Appendix}{Appendices}
\Crefname{appendix}{Appendix}{Appendices}

\newcommand{\er}{\operatorname{er}}                 
\newcommand{\crit}{\operatorname{crit}}                
\newcommand{\wt}{\operatorname{wt}}           
\newcommand{\wvar}[1]{\mathsf{#1}}             
\newcommand{\ctr}{\kappa}                     

\newcommand{\nint}[1]{\lVert #1\rVert}        
\newcommand{\Lab}{\operatorname{CR}}         
\newcommand{\pre}{\operatorname{pre}}
\newcommand{\Labeling}{\textsc{Labeling}}
\newcommand{\GreedyCover}{\textup{\textsc{GreedyAllocation}}}
\newcommand{\FGA}{\operatorname{MinGA}}
\newcommand{\TER}{\textup{\textsc{Tree-Edge-Ranking}}}
\newcommand{\AMtwoCVP}{\textup{\textsc{AM2CVP}}}

\newcommand{\sigH}{\mathsf H}                 
\newcommand{\sigL}{\mathsf L}                 
\newcommand{\Pc}{\mathrm{P}}
\newcommand{\NC}{\mathrm{NC}}

\def\DEBUGFLAG{}
\ifdefined\DEBUGFLAG

  \def\rem#1{{\marginpar{\raggedright\scriptsize #1}}}
  \newcommand{\mck}[1]{\rem{\textcolor{red}{$\bullet$ #1}}}
  \newcommand{\mjg}[1]{\rem{\textcolor{blue}{$\bullet$ #1}}}
  \newcommand{\mclaudecode}[1]{\rem{\textcolor{teal}{$\bullet$ #1}}}
  \newcommand{\mcodex}[1]{\rem{\textcolor{Plum}{$\bullet$ #1}}}

\else

  \newcommand{\mck}[1]{}
  \newcommand{\mjg}[1]{}
  \newcommand{\mclaudecode}[1]{}
  \newcommand{\mcodex}[1]{}
\fi

\colorlet{heavycol}{NavyBlue}
\colorlet{attcol}{BrickRed}
\colorlet{grpA}{NavyBlue}
\colorlet{grpB}{BurntOrange}
\colorlet{gapcol}{ForestGreen}
\colorlet{hicol}{NavyBlue}
\colorlet{locol}{BurntOrange}
\colorlet{visc}{BrickRed}
\colorlet{forwardcol}{NavyBlue}
\colorlet{delivercol}{ForestGreen}
\tikzset{
  terpforward/.style={-{Stealth[length=2mm]},thick,solid,draw=forwardcol},
  terpdeliver/.style={-{Stealth[length=2mm]},thick,densely dashed,draw=delivercol}
}

\title{Parallel Edge Ranking of Trees}
\author{Jeff Giliberti \quad MohammadTaghi Hajiaghayi \quad Changki Yun\\[0.3em]
\normalsize University of Maryland, College Park\\
\small\texttt{\{jeffgili,hajiagha,ckyun\}@umd.edu}}
\date{}

\begin{document}
\maketitle\vspace{-2.5em}

\begin{abstract}
In this work, we prove that computing the edge ranking of a tree in parallel is P-Complete, and complement this hardness with new algorithmic results for both PRAM and the MPC model. An optimal tree edge ranking assigns positive integer ranks to the edges such that any two edges with the same rank are separated by an edge of higher rank, while minimizing the highest rank. Tree edge ranking abstracts several classical problems such as edge-query binary search on trees, minimum-height dendrograms, reversible pebble game, and parallel assembly in manufacturing. Its parallel complexity has remained open for almost thirty-five years, since the seminal work of de la Torre, Greenlaw, and Sch{\"a}ffer~[SODA'93], and was listed as an open problem in the book \textit{Limits to Parallel Computation} by Greenlaw, Hoover, and Ruzzo~(1995). 
Specifically, we show that deciding whether a tree has edge ranking at most $K$ is P-complete, already for trees of diameter at most six. Our reduction encodes monotone circuit evaluation into a greedy allocation procedure underlying known sequential algorithms. Using a rooted product of trees, we further prove that computing an edge ranking within constant additive error remains $\Pc$-hard.
 
 On the algorithmic side, we break the long-standing 2-approximation barrier by designing a deterministic NC algorithm that computes an $(1+o(1))$-approximate edge ranking. Second, we show that P-completeness can be bypassed in Massively Parallel Computation (MPC) with strongly sublinear local memory: we give a deterministic algorithm that computes an optimal edge ranking of an $n$-vertex, diameter-$D$ tree in $O(\log D + \frac{n^{3/4}D^{1/4}}{S})$ rounds with $S = \Theta(n^{3/4+\delta})$ local memory, for $\delta \in (0, 1/4)$, requiring only $O(1)$ rounds on the hard diameter-six instances. Our MPC algorithm circumvents the $\Omega(n)$-memory bottleneck of previous methods by compressing the information required to produce an optimal ranking from subtree merges. Overall, this result reinforces the separation between $\mathsf{NC}$ and what can be computed efficiently in MPC.
\end{abstract}

\setcounter{tocdepth}{1}
\tableofcontents
\newpage
\section{Introduction}\label{sec:intro}

An \emph{edge ranking} of a tree assigns positive integer ranks to its edges so that the path between any two distinct edges of equal rank contains an edge of higher rank. The tree edge ranking problem asks for a ranking that minimizes the maximum rank used.
Computing an optimal edge ranking on trees is a fundamental combinatorial problem, with numerous applications, that has received considerable attention since the early nineties~\cite{IyerRV91}, and even earlier~\cite{linial1985searching, bennett1989time} in slightly different forms, later shown to be closely related \cite{D08, KomarathSS18}. Beyond its structural interest, tree edge ranking encompasses a broad variety of problems.
In hierarchical clustering, it defines a minimum-height dendrogram of connected subtrees~\cite{IyerRV91, NWD89}. 
In trees and tree-like partially ordered sets, it determines the minimum worst-case number of edge queries required to locate a target element~\cite{linial1985searching, BFN99,OnakP06,MozesOW08,D08}. In combinatorial games, it is equivalent to finding the reversible pebble game number of a tree~\cite{bennett1989time,Chan13,KomarathSS18}. In practical scenarios, it models assembling a product from its components in the minimum number of stages~\cite{IyerRV91}, evaluating parallel database queries~\cite{MakinoUI01,DereniowskiK06}, and synchronizing distributed file systems~\cite{BFN99}.

Tree edge ranking was initially suspected to be NP-hard, owing to its similarity to certain NP-complete edge-labeling problems \cite{perl1982complexity} (see also \cite{delaTorreGS95}). Further, Iyer, Ratliff, and Vijayan~\cite{IyerRV91} gave a polynomial-time 2-approximation algorithm and posed the existence of an exact polynomial-time algorithm as an open question. In a breakthrough result, de~la~Torre, Greenlaw, and Sch{\"a}ffer~\cite{delaTorreGS93,delaTorreGS95} placed the problem in P. Their runtime was later improved in a series of works \cite{ZhouNishizeki1994, ZhouNishizeki1995}, culminating in the optimal linear-time and space algorithm of Lam and Yue~\cite{LamY98, LamY01}.

\paragraph{Hardness of edge ranking.} In the parallel setting, progress on this problem remained elusive. In particular, its $\NC$ membership (whether an optimal tree edge ranking can be found in polylogarithmic time with polynomially many processors) has remained wide open. An $\NC$ algorithm that finds a ranking with additive error $O(\log n)$ is implicit in the work of~\cite{IyerRV91}. A parallel algorithm for exact ranking with an exponential dependency on the maximum degree was developed in~\cite{delaTorreGS95}. In contrast, an algorithmic routine, \textsc{FillSlots}, used by the sequential algorithm of \cite{delaTorreGS95} was shown to be $\Pc$-complete, ruling out a direct parallelization of their approach, but without settling the complexity of tree edge ranking itself. Accordingly, the book \emph{Limits to Parallel Computation} by Greenlaw, Hoover, and Ruzzo, regarded as the P-completeness analog of Garey and Johnson's NP-completeness book~\cite{garey2002computers}, lists tree edge ranking as an open problem \cite[Problem B.1.2]{GreenlawHR95}. While edge ranking on \textit{general} graphs was subsequently shown to be NP-hard \cite{LamY98hard}, determining whether tree edge ranking lies in $\NC$ or is $\Pc$-complete has stood unresolved.

We resolve this almost 35-year-old open question, proving that tree edge ranking is $\Pc$-complete:

\begin{theorem}\label{thm:pcomplete} Given a tree $T$ and an integer $K \ge 0$, deciding whether $T$ has edge ranking at most $K$ is $\Pc$-complete under $\NC^1$ many-one reductions, already for trees of diameter at most six.
\end{theorem}

Moreover, we show $\Pc$-hardness even for computing an edge ranking within constant additive error. These results strongly suggest hardness of efficient parallelization, meaning that this problem has an inherently sequential structure. In particular, they rule out $\NC$ algorithms in shared-memory models such as PRAM, the standard abstraction for shared-memory parallel architectures~\cite{FortuneW78, KarpR90, Jaja92}. However, they neither rule out $\NC$ algorithms with super-constant, sublogarithmic additive error, nor transfer directly to other parallel computational settings, such as the model of Massively Parallel Computation (MPC)~\cite{KarloffSV10, GoodrichSZ11, BeameKS17, AndoniNOY14}, where coarse-grained parallelism is interleaved with sequential, local computation across machines with limited memory.

\paragraph{Parallel algorithms.} We make considerable progress in both directions, giving an algorithmic result for both PRAM and MPC. First, we show that tree edge ranking admits better-than-$O(\log n)$ additive approximation:
\begin{theorem} \label{thm:approx-nc-main}
    Given an $n$-vertex tree, there is a deterministic CRCW PRAM algorithm that finds an edge ranking with additive error at most $\frac{\log_2 n}{\log_2\log_2 n}$, in $O(\log^4 n)$ time with polynomial work.
\end{theorem}
Our result for general graphs is the first improvement in thirty-five years since the work of~\cite{IyerRV91}. Moreover, as the maximum rank is at least $\log_2 n$ for any tree, our algorithm has multiplicative approximation factor $(1 + o(1))$. See \Cref{tab:parallel-comparison} for a comparison of known results.

\begin{table}[htbp]
\centering
\small
\renewcommand{\arraystretch}{1.25}
\setlength{\tabcolsep}{10pt}
\begin{tabular}{@{}llll@{}}
\toprule
Result & Tree Class & Maximum Rank & Time (polynomial work) \\
\midrule
\cite{IyerRV91}       & Arbitrary       & $\le \er(T) + \log_2 n$                                     & $O(\log^2 n)$ \\
\cite{delaTorreGS95}  & Bounded-degree  & Optimal                                                     & $O(\log^2 n \log \log n)$ \\
\cite{KashemZN98}     & Bounded-degree  & Optimal                                                     & $O(\log n)$ \\
\midrule
This work             & Arbitrary       & $\displaystyle\le \er(T) + \frac{\log_2 n}{\log_2 \log_2 n}$ & $O(\log^4 n)$ \\
\bottomrule
\end{tabular}
\caption{Parallel algorithms for edge ranking on $n$-vertex trees; $\er(T)$ denotes
    optimal maximum rank.}
\label{tab:parallel-comparison}
\end{table}

\textbf{MPC algorithm.}
Our second direction investigates whether the barrier of $\Pc$-completeness can be bypassed in modern settings of parallel computations. Over the past two decades, Massively Parallel Computation (MPC) has emerged as the standard abstraction for large-scale parallel computation, capturing prominent frameworks such as MapReduce~\cite{DeanG04,DeanG08}, Spark~\cite{ZahariaCFSS10,ZahariaCDRRSS12}, Hadoop~\cite{white2012hadoop}, and Dryad~\cite{IsardBYBF07}. In this model, the input of size $N$ is evenly distributed across $M$ machines, each with local memory $S$ such that $M \cdot S = \widetilde O(N)$.\footnote{We use the notation $\widetilde O(\cdot)$ to suppress polylogarithmic factors.}
Computation proceeds in synchronous rounds where machines perform local computation and all-to-all communication, subject to the constraint that the total data sent and received by any machine per round is $O(S)$. The primary objective is to minimize the number of  rounds. Although local computation is not charged to the round complexity, the degree of parallelism is significantly restricted by the local memory constraint. 
Indeed, with a strongly sublinear local memory, no single machine can see the entire graph and global decisions must be made using limited, local information.

A primary motif of MPC research has been obtaining faster algorithms than their PRAM counterparts, with the gold standard of a sublogarithmic or constant number of rounds where PRAM requires $\Omega(\log n)$ time. Conversely, under conditional barriers such as the widely believed 1-vs-2 cycle conjecture~\cite{RoughgardenVW18, NanongkaiS22}, basic problems like connectivity resist sublogarithmic MPC rounds \textit{while} admitting logarithmic-time PRAM algorithms~\cite{ShiloachV82,LTZ20}. This suggests that certain problems \textit{cannot} gain substantial MPC speedups.

The reverse direction of whether problems that are \emph{intractable} in classical parallel models, i.e., P-Complete, can be parallelized efficiently in MPC (with strongly sublinear local memory) has remained largely unexplored. 
$\Pc$-completeness reductions based on circuits neither rule out polylogarithmic time MPC algorithms nor translate to efficient MPC algorithms, meaning that problems equivalent under P-Complete reductions may not share the same MPC complexity class. 
To our knowledge, the only known instance where a $\Pc$-complete problem admits an efficient MPC algorithm is the work of~\cite{DhulipalaDKPSS20} on adaptive batch-dynamic connectivity where, however, batch updates are constrained in size by the machine's local memory. It is still poorly understood which $\Pc$-complete problems exhibit a sufficiently \emph{local} structure to admit efficient MPC parallelization.

Our work shows that exact tree edge ranking exhibits an interesting structure: despite being $\Pc$-complete under fine-grained parallelization, its combinatorial structure can be compressed to yield an $O(1)$-round MPC algorithm with strongly sublinear local memory\footnote{We say $f(n)$ is strongly sublinear if $f(n) = O(n^{1-\varepsilon})$ for some fixed $\varepsilon > 0$; this is slightly different from the notion of fully-scalable strongly-sublinear MPC algorithms where $\varepsilon > 0$ is a free parameter.} for trees of constant diameter, that is, the very hard instances used to prove P-completeness in \cref{thm:pcomplete}. Our result is more general in that we show a trade-off between local memory and round complexity based on the tree diameter.

\begin{theorem}\label{thm:mpc}
For any fixed $\delta > 0$, an optimal edge ranking of an $n$-vertex tree of diameter $D$ can be computed deterministically in MPC using $S = \Theta(n^{3/4+\delta})$ local memory, $O(n)$ total memory, and
\[
    O\!\left(\frac{n^{3/4} D^{1/4}}{S} + \log D\right) \text{ rounds.}
\]
In particular, $O(n^{3/4+\delta})$ local memory yields $O(1)$ rounds on trees of constant diameter.
\end{theorem}

We discuss further related work in \cref{sec:relwork}.
\section{Technical overview}\label{sec:overview-alternative}
This section presents an overview of our main techniques: \cref{sec:ov-pcom} covers the $\Pc$-completeness reduction, and \cref{sec:ov-alg} outlines the intuition behind our PRAM and MPC algorithms.

\subsection{Hardness of Tree Edge Ranking}\label{sec:ov-pcom}
Our reduction is from the alternating monotone circuit value problem with bounded fan-in and fan-out (\AMtwoCVP{}), a known $\Pc$-complete problem. It goes through an intermediate allocation problem, \GreedyCover{}, and then to tree edge ranking (\cref{sec:pcomplete}).

To simplify exposition, we start with a simpler case: a rooted tree in which all edges except those incident to the root have already been ``optimally'' ranked. We are left with the task of assigning ranks to these root-incident edges to preserve ``optimality''.

We define optimality using visible ranks. A rank is \emph{visible} if no larger rank lies between an edge carrying it and the root. Only ranks visible from a subtree's root can conflict with ranks outside it. We compare visible sets (sets of visible ranks) by the sum of their \emph{weights}, where a visible rank $a$ adds $2^{a-1}$ to the total weight. Note that comparison between total weights reflects the lexicographic order of the visible sets; hence, minimizing total weight also minimizes the maximum rank. The minimum-weight visible set of a rooted tree $T$ is called its \emph{critical set} $\crit(T)$, with \emph{critical weight} $\wvar{w}(T)$, and a ranking attaining it is \emph{critical}. Therefore, finding a critical ranking, i.e., achieving a ranking of minimum weight, also provides an optimal ranking.

Moreover, assuming that child subtrees are critically ranked, one can find a critical ranking by ranking the root-incident edges without modifying any rankings inside the subtrees. Referred to as \emph{labeling}, this procedure is crucial to sequentially computing an optimal ranking in the algorithms of~\cite{delaTorreGS95,LamY01}.

The reduction is based on showing that deciding whether the critical weight of the entire tree exceeds $2^{K} -1$ is P-complete, even when a critical ranking of each subtree is provided. This is clearly a subproblem of determining the existence of a tree edge ranking with rank at most $K$.

\paragraph{\TER{} and \GreedyCover{}.}
We translate the labeling problem, i.e., critically ranking root-incident edges, into a combinatorial allocation problem. Given a tree edge ranking, each visible rank is visible in a unique \emph{branch}, a child subtree together with its edge incident to the root. Thus, $[K]$ supports a valid ranking of the tree if and only if it can be partitioned among the branches so that each branch admits a valid ranking whose visible set is contained in its part.

Although certifying each branch may appear to repeat the same question, it simplifies because of the branch's root-incident edge, which hides all smaller visible ranks in its subtree: a branch $B$ admits a valid ranking with visible set contained in $X$ if and only if the weight of $X$ is at least $\wvar{w}(B)$. Indeed, if $C=\crit(B)$ is not already contained in $X$, the highest rank $c$ where $X$ and $C$ differ lies in $X\setminus C$, and raising the root-incident edge of $B$ to $c$ hides all lower ranks of $C$ while matching $X$ above $c$. We therefore define the \GreedyCover{} problem: given an integer $K \ge 1$ and $d \ge 1$ positive integer \emph{requirements} $\wvar{q}_1, \ldots, \wvar{q}_d$, decide whether $\{2^{K-1}, \ldots, 2^{0}\}$ can be partitioned (allocated) into $d$ subsets whose weights sum to at least the corresponding requirements. This problem captures \TER{} at the root when each requirement is the critical weight of the corresponding branch.

Conversely, in \cref{sec:demands}, we show that an arbitrary \GreedyCover{} instance reduces to \TER{}. Note that a star graph of $k$ edges, rooted at an arbitrary leaf, has critical weight $2^{k-1}$. Moreover, identifying the roots of stars of $k_1 > \cdots > k_t > 0$ edges yields critical weight
\[
2^{k_1 - 1} + \cdots + 2^{k_t - 1}.
\]
Adding an edge above this root to turn the tree into a branch increases its critical weight by one: assigning that edge the smallest rank absent from the critical set turns on the lowest $0$-bit of the weight and hides all smaller visible ranks, clearing every trailing $1$-bit just as in binary addition. Therefore, an arbitrary $\ell$-bit requirement $\wvar{q}$ is realized as the critical weight of an $O(\ell^2)$-vertex branch of height at most three via the binary expansion of $\wvar{q}-1$. Identifying the roots of these branches yields an equivalent \TER{} instance of diameter at most six, constructed in $\NC$ with polynomial blowup.

\paragraph{\GreedyCover{} and \AMtwoCVP{}.}
We establish $\Pc$-completeness by encoding circuit evaluation in the greedy procedure for \GreedyCover{}.
At \emph{clock} $e=K-1,\ldots,0$, the procedure gives weight $2^e$ to the largest residual requirement.
A residual below $2^e$ is covered and disappears; otherwise, subtracting $2^e$ leaves the rest for later clocks.
If some residual is at least $2^{e+1}$, or two are at least $2^e$, it rejects because the remaining weights cannot meet them; a \GreedyCover{} instance is feasible if and only if this procedure never rejects.

To see how an input bit can determine what survives, start at clock $8$, let $1\le a,b<2^4$, and fix three circuit requirements
\[
  2^7+a,\qquad 2^6+2^4,\qquad 2^4+b.
\]
Encode `true' by the additional requirement $2^8$ and `false' by the two requirements $\{2^5+1,2^5\}$.
\Cref{fig:overview-signal} shows the two runs.
On true, the signal consumes $2^8$, so the next weight removes only the leading bit of $2^7+a$.
On false, $2^8$ instead covers that entire requirement, including $a$.
The later clocks reverse the roles for $2^4+b$, covering it on true but subtracting $2^4$ and leaving $b$ on false.
The extra requirement $2^6+2^4$ absorbs the weight $2^4$ on true, preventing it from covering $a$.
After clock $4$, only $a$ remains on true, and only $b$ on false. Thus $a$ and $b$ are chosen once, as part of the circuit, while the signal decides which one survives and is passed to later clocks.

In \cref{sec:am2-to-ga}, we lift this mechanism to simulate a topologically ordered \AMtwoCVP{} circuit by dedicating a fixed $18$-clock block to each node. Generalizing the example above gives a $15$-clock \emph{selector} that converts a node's output signal into input \emph{messages} for its outgoing wires; because fan-out is at most two, a constant number of message slots suffices per selector. Evaluating a gate is subtle because the greedy procedure consumes residuals destructively and two incoming messages do not add, yet the same subtract-versus-cover shift yields a $3$-clock AND-gate gadget: at clocks $e,e-1,e-2$, the right input $y$ supplies two requirements with leading bits $2^{e-1}$ and $2^{e-2}$ shielding an embedded $y$-signal below $2^{e-3}$. When the left input encodes $x=1$ ($2^e$), it consumes clock $e$, so clocks $e-1$ and $e-2$ subtract those leading bits and expose $y$'s signal; when $x=0$ (a guard at $2^{e-3}+2^{e-4}$ plus a $0$-signal), clocks $e$ and $e-1$ instead cover and erase both right-input requirements while clock $e-2$ clears the guard, leaving the $0$-signal. Because AND and OR gates alternate, having every selector transmit the \emph{complemented} bit lets De Morgan's law reduce every gate to this single AND gadget, followed by a final test requirement that causes rejection precisely when the circuit evaluates to $0$.
\begin{figure}[htbp]
\centering
\begin{tikzpicture}[
  toyallocated/.style={circle,draw=black,fill=black,inner sep=0pt,
    minimum size=3.8pt,line width=.2pt},
  toycovered/.style={circle,draw=black!55,draw opacity=.6,
    fill=white,inner sep=0pt,minimum size=4pt,line width=.6pt},
  toylabel/.style={font=\footnotesize,anchor=east,inner sep=0pt},
  toyhead/.style={font=\footnotesize,inner sep=0pt},
  toytail/.style={draw=delivercol!75,fill=delivercol!9,
    text=delivercol,rounded corners=.5pt,font=\footnotesize,
    inner xsep=3pt,inner ysep=1pt,minimum width=.65cm},
]
\newcommand{\toybit}[3]{%
  \node[#3] at ({1.75+.66*(#1)},{-.48*(#2)}) {};}
\newcommand{\toytailbox}[3]{%
  \node[toytail,opacity=#3] at (5.25,{-.48*(#1)}) {$#2$};}
\newcommand{\toyaxes}[1]{%
  \node[font=\small] at (3.3,1.02) {#1};
  \foreach \col/\power in {0/8,1/7,2/6,3/5,4/4}
    \node[toyhead] at ({1.75+.66*\col},.48) {$2^{\power}$};
  \node[toyhead] at (5.25,.48) {$<2^4$};
  \draw[black!20] (4.82,.18) -- (4.82,-2.13);
  \draw[black!15] (.15,-.73) -- (5.65,-.73);
  \node[toylabel] at (1.42,-.96) {$2^7+a$};
  \node[toylabel] at (1.42,-1.44) {$2^6+2^4$};
  \node[toylabel] at (1.42,-1.92) {$2^4+b$};}

\begin{scope}
\toyaxes{true}
\node[toylabel] at (1.42,0) {$2^8$};
\toybit{0}{0}{toyallocated}
\toybit{1}{2}{toyallocated}
\toytailbox{2}{a}{1}
\toybit{2}{3}{toyallocated}
\toybit{4}{3}{toyallocated}
\toybit{3}{4}{toyallocated}
\toybit{4}{4}{toycovered}
\toytailbox{4}{b}{.4}
\node[font=\small,text=delivercol] at (3.3,-2.5) {$\longrightarrow\ \{a\}$};
\end{scope}

\begin{scope}[xshift=7.2cm]
\toyaxes{false}
\node[toylabel] at (1.42,0) {$2^5+1$};
\node[toylabel] at (1.42,-.48) {$2^5$};
\toybit{2}{0}{toyallocated}
\toybit{3}{0}{toycovered}
\node[font=\footnotesize,text=black!45] at (5.25,0) {$1$};
\toybit{3}{1}{toyallocated}
\toybit{0}{2}{toyallocated}
\toybit{1}{2}{toycovered}
\toytailbox{2}{a}{.4}
\toybit{1}{3}{toyallocated}
\toybit{2}{3}{toycovered}
\toybit{4}{3}{toycovered}
\toybit{4}{4}{toyallocated}
\toytailbox{4}{b}{1}
\node[font=\small,text=delivercol] at (3.3,-2.5) {$\longrightarrow\ \{b\}$};
\end{scope}
\end{tikzpicture}
\caption{The same circuit requirements leave $a$ or $b$ after clocks $8,\ldots,4$.
Filled marks show allocated weights; faded marks and entries are covered.}
\label{fig:overview-signal}
\end{figure}
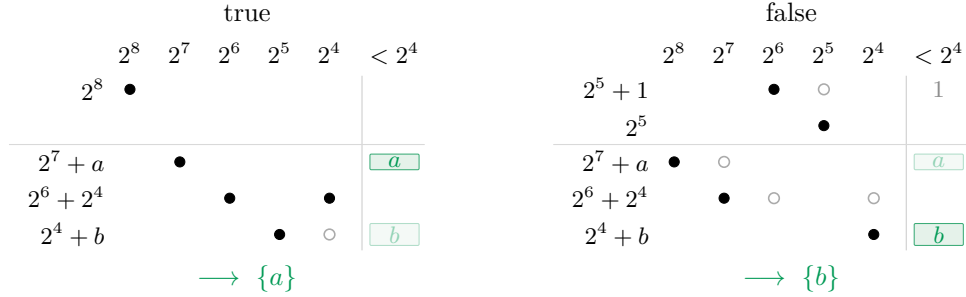

\paragraph{Constant additive error.}
To extend our $\Pc$-hardness result to any constant additive error (\cref{sec:approx-inapprox}), we amplify the gap in $\er(T)$ via tree composition. While graph composition was used earlier to study edge rankings of general graphs~\cite{Yue96}, we show that the \emph{rooted product} of trees~\cite{Godsil_McKay_1978} provides a clean additive structure on trees matching our purpose. For rooted trees $T$ and $S$, form $T\odot S$ by attaching a copy of $S$ at every vertex of $T$, identifying that vertex with the copy's root. We prove that
\[
 \er(T\odot S)=\er(T)+\er(S).
\]
The upper bound is immediate: rank each copy of $S$ optimally, and increase every rank in an optimal ranking of $T$ by $h:=\er(S)$. The lower bound follows from a stronger identity on critical weights:
\[
\wvar{w}(T \odot S) = 2^{h}\wvar{w}(T) + \wvar{w}(S).
\]
To see why no ranking can do better at the root of $T\odot S$, note that inductively each branch $B_i$ of $T$ induces a branch in $T\odot S$ with requirement $2^h(\wvar{w}(B_i)-1)+\wvar{w}(S)+1$, while the root copy of $S$ requires $\wvar{w}(S)$. Because ranks above $h$ have weights divisible by $2^h$ and ranks in $[h]$ sum to $2^h-1$, a branch $i$ can receive less than $2^h\wvar{w}(B_i)$ from ranks above $h$ only if it takes at least $\wvar{w}(S)+1$ from $[h]$. Since $\wvar{w}(S)\ge 2^{h-1}$, branch $i$ and the root copy of $S$ cannot both fit in $[h]$, so the root copy of $S$ must hold some rank above $h$; transferring that rank to branch $i$ yields a valid allocation for $T$ above $h$ plus weight exceeding $\wvar{w}(S)$ in $[h]$.

Repeatedly applying this product amplifies the unit gap of our $\Pc$-completeness reduction to any constant integer $r \ge 1$. Indeed, for our hard tree $T$ of radius at most $3$, the $r$-fold rooted product $T^{\odot r} := T \odot \cdots \odot T$ has $|V(T)|^r = |V(T)|^{O(1)}$ vertices, diameter at most $6r$, and $\er(T^{\odot r}) = r\er(T)$. Thus, for any threshold $K \ge 0$,
\[
\er(T) \le K \iff \er(T^{\odot r}) \le rK,\qquad \er(T) > K \iff \er(T^{\odot r}) \ge rK + r,
\]
so any $\NC$ algorithm with additive error less than $r$ on $T^{\odot r}$ decides whether $\er(T) \le K$, ruling out constant additive error in $\NC$ unless $\Pc=\NC$.

\subsection{Parallel Algorithms}\label{sec:ov-alg}
We show that the hardness of tree ranking can be overcome in two complementary ways: in $\NC$ by relaxing exactness to a superconstant, sublogarithmic additive error (\cref{sec:approx}), and in MPC by compressing the information needed for exact subtree merges (\cref{sec:mpc}). Both algorithms build on a common structural primitive, which we call \emph{size-threshold decomposition} (\cref{lem:size-threshold}). This primitive can be viewed as a generalization of the separator approach of~\cite{IyerRV91}. Given an $N$-vertex rooted tree and a branching parameter $b\ge 2$, retaining all vertices whose descendant subtrees have size at least $N/b$ produces a connected \emph{retained tree} $T_0$ with at most $b$ leaves (whose subtrees are disjoint), while partitioning the rest of the tree into disjoint \emph{attached subtrees} of size less than $N/b$ hanging from $T_0$. The attached subtrees are small enough to be ranked in parallel, reducing each step to extending their rankings across the retained tree $T_0$.

\subsubsection{Approximation in NC}
Ranking the attached subtrees recursively and combining them optimally across $T_0$ requires their full critical sets, which are $\Pc$-complete to compute in parallel. In contrast, shrinking $T_0$ to a single separator vertex (e.g. centroid) as in~\cite{IyerRV91} ($b=O(1)$) avoids computing critical sets by losing one rank per level, but requires $\Theta(\log n)$ recursion levels and thus incurs $O(\log n)$ additive error. To reduce the recursion depth to $O(\log_b n)$ for a larger shrinkage factor $b$, while still losing at most one rank per level, we introduce a \emph{lower-bound replacement} step: after ranking each attached subtree $T_i$ with maximum rank $h_i$, we contract $T_i$ into its root, turning its edge to $T_0$ into a \emph{leaf edge} incident to $T_0$ with lower bound $h_i+1$. Because a rank of at least $h_i+1$ on this leaf edge exceeds every rank in $T_i$, it shields $T_i$ from the rest of the tree, so any valid ranking of the contracted tree that respects the lower bounds extends directly to $T$ and loses at most one rank across the recursion level. The problem at each level is then to compute a \emph{lower-bounded critical ranking} of $T_0$ together with its attached leaf edges. Since $T_0$ has at most $b$ leaves, its root, leaves, and vertices with at least two children in $T_0$ partition $T_0$ into $O(b)$ paths. Using a logic similar to the well-known rake-and-compress, we can rank $T_0$ bottom-up in $O(b)$ steps by alternating two $\NC$ subroutines: combining branches at a path endpoint and extending a ranking along a path.

\paragraph{Combining branches at a path endpoint.}
At a path endpoint $u\in T_0$, we are given the critical rankings of its $d\le b$ child branches $B_1,\ldots,B_d$ in $T_0$ and the lower bounds on the leaf edges attached at $u$. To rank the edges incident to $u$ with minimum visible weight, each branch $B_i$ needs a subset of ranks of total weight at least its critical weight $\wvar{w}(B_i)$, each leaf edge needs a single rank at least its lower bound, and all assigned ranks must be distinct. Searching directly over the ranks of the $d$ branches in $T_0$ would take $n^{\Omega(d)}$ work, which is superpolynomial for $d=\Theta(\log n)$. Instead, we build a dynamic programming table that scans candidate ranks $h$ from $2n$ down to $1$ (since $n$ edges with lower bounds at most $n$ never need ranks of $2n$ or higher). Each DP entry $(h,j,\mathcal{B})$ is indexed by the current rank $h$, the index $j$ of the next leaf edge to rank, and the subset $\mathcal{B}\subseteq[d]$ of branches in $T_0$ that are not yet satisfied; it stores the minimum weight needed to complete the allocation using ranks in $[h]$.

At step $h$, we either leave rank $h$ unused, assign it to the $j$-th leaf edge (if $h$ exceeds its lower bound), or assign it to an unfinished branch $B_i\in\mathcal{B}$. Sorting the leaf edges in decreasing order of their lower bounds ensures that we can always process them in order, so the single index $j$ suffices to track which leaf edges remain. Likewise, the subset $\mathcal{B}$ suffices to track the $d$ branches in $T_0$ without recording the weight each branch has accumulated from ranks above $h$. This is because every rank above $h$ has weight a multiple of $2^h$, whereas all remaining ranks $1,\ldots,h$ together sum to $2^h-1$. Among all multiples of $2^h$, there is at most one value from which a branch $B_i$ is neither already satisfied (at least $\wvar{w}(B_i)$) nor impossible to complete (more than $2^h-1$ short of $\wvar{w}(B_i)$). Thus, knowing that $B_i\in\mathcal{B}$ is still unfinished at step $h$ fixes its accumulated weight and determines whether rank $h$ finishes it. Since there are only $2^d n^{O(1)}\le 2^b n^{O(1)}$ states $(h,j,\mathcal{B})$, computing a minimum-weight allocation at $u$ reduces to a single-source shortest-path computation on a DAG of size $2^{O(b)}n^{O(1)}$, taking $O(\log^2 n)$ time and $2^{O(b)}n^{O(1)}$ work.

\paragraph{Extending a ranking along a path.}
Along a path $u=v_0,\ldots,v_m=v$ in $T_0$, we are given the critical ranking of $T_v$ at the bottom endpoint $v$ and lower bounds on the leaf edges attached along the path. Our goal is to rank the path and leaf edges to minimize the visible weight at the top endpoint $u$. Since the path can have length $m=\Omega(n)$, advancing vertex by vertex along the path is too slow in $\NC$. Instead, we build a dynamic programming table over subpaths $[s,t]=v_s,\ldots,v_t$ and rank limits $h\le 2n$ (which, as before, bounds all ranks needed by $n$ edges with lower bounds at most $n$). Each entry $([s,t],h)$ stores the minimum visible weight at $v_s$ of a valid ranking of $[s,t]$ and its leaf edges using ranks in $[h]$ (ignoring leaf edges with bounds above $h$).

To compute $([s,t],h)$, the recurrence guesses the path edge $v_iv_{i+1}$ on $[s,t]$ that receives the largest rank $k\le h$. Once $k$ is fixed, we greedily assign the leaf edges on $[s,t]$ whose lower bounds (or visible ranks from $T_v$) exceed $k$ the smallest available distinct ranks in $(k,h]$. Indeed, because no path edge has rank above $k$, no two edges on $[s,t]$ can share a rank in $(k,h]$. Moreover, a leaf edge is ranked above $k$ if and only if its bound exceeds $k$, since otherwise swapping its rank with $k$ would raise the maximum path-edge rank without increasing the visible weight. Removing the path edge $v_iv_{i+1}$ of rank $k$ then splits $[s,t]$ into an upper subpath $[s,i]$ (closer to $v_s$) and a lower subpath $[i+1,t]$ (closer to $v_t$), both restricted to ranks at most $k-1$. Because rank $k$ on $v_iv_{i+1}$ hides all smaller ranks in the lower subpath from $v_s$, the lower subproblem $([i+1,t],k-1)$ only needs to be \emph{feasible} (admit some valid ranking), whereas the visible weight below $k$ at $v_s$ comes entirely from the upper subproblem $([s,i],k-1)$.

Evaluating this $O(n^3)$-state table by increasing subpath length would still take $\Omega(m)$ sequential steps, so we evaluate it in two parallel phases of $O(\log^2 n)$ time. First, $O(\log n)$ doubling rounds of graph reachability identify all feasible states $([i+1,t],k-1)$. Once feasibility of the lower subproblem is known, each transition depends on only a single subproblem $([s,i],k-1)$, which reduces computing the optimal weights to a single-source shortest-path computation on an $O(n^3)$-vertex DAG.

\smallskip
Alternating these two subroutines bottom-up across the $O(b)$ paths and endpoints of $T_0$ solves each recursion level in $O(b\log^2 n)$ time and $2^{O(b)}n^{O(1)}$ work. Choosing $b=\lceil\log_2 n\rceil$ keeps the work polynomial while bounding both the recursion depth and the total additive error by $\lfloor\log_b n\rfloor\le \log_2 n/\log_2\log_2 n$ in $O(\log^4 n)$ total time.

\subsubsection{Exact ranking in MPC}
In MPC, local computation allows a single machine to compute an exact critical ranking of a subtree once the required information fits into its local memory $S=\Theta(n^{3/4+\delta})$. In an $N$-vertex subtree whose attached subtrees have already been critically ranked, extending their rankings across a $k$-vertex retained tree $T_0$ naively requires gathering all their critical sets, which can take $\Omega(N)$ space. We overcome this memory bottleneck in two steps. First, we design a \emph{compression and batched labeling scheme} that encodes $T_0$ and all attached critical sets in $O(N^{3/4}k^{1/4})$ space, and computes the exact critical ranking directly on this compressed representation within the same space bound. Second, we use a multi-stage bottom-up decomposition with carefully chosen size thresholds so that each stage's subproblems fit into $O(S)$ memory and can be solved in parallel across machines.

\paragraph{Compressing critical sets.}
To compress the critical sets of the attached subtrees, we start from Lam and Yue's \emph{segment representation}~\cite[Lemma~5.1]{LamY01}, which stores the critical set of an $s$-vertex subtree in $O(\sqrt{s})$ words as a list of maximal intervals of consecutive ranks (\emph{segments}). When many small subtrees are attached to $T_0$, however, even their segment representations can sum to $\Omega(N)$ words. To compress them further, we group subtrees attached to the same retained vertex $v\in T_0$ by the maximum rank $\mu$ of their critical sets. Within each group $(v,\mu)$, we keep \textit{one set of minimum weight} in segment representation, and replace all other branches by their count, i.e., we treat each as a \emph{star branch of type $\mu$} whose critical set is $[\mu]$. Indeed, because the visible sets of the branches at $v$ must be pairwise disjoint, at most one branch in a $(v,\mu)$-group can receive a rank at most $\mu$ on its edge to $v$ (which leaves $\mu$ visible), and an exchange argument shows that an optimal ranking at $v$ can always choose a minimum-weight branch in the group as its sole exception. Every other branch in the group must receive a rank strictly above $\mu$ on its edge to $v$, making it indistinguishable from a branch with critical set $[\mu]$.

We next bound the total size of this compressed representation across all $k$ retained vertices of $T_0$ together. A subtree whose critical set has maximum rank $\mu$ must contain at least $\mu$ edges. Consequently, $g_v$ attached subtrees at $v$ with distinct maximum ranks $\mu_1<\cdots<\mu_{g_v}$ contain at least $1+2+\cdots+g_v=\Omega(g_v^2)$ vertices in total. Thus a retained vertex $v$ whose attached subtrees have $N_v$ vertices in total has at most $g_v=O(\sqrt{N_v})$ groups. Since the subtrees attached to the $k$ retained vertices are disjoint and have total size at most $N$, applying Cauchy--Schwarz across them bounds the total number of explicit critical sets kept across $T_0$ by $O(\sqrt{kN})$. Applying Cauchy--Schwarz a second time to their individual $O(\sqrt{s})$ segment bounds leads to the space bound of $O(\sqrt{N\sqrt{kN}})=O(N^{3/4}k^{1/4})$.

\paragraph{Batched labeling on compressed sets.}
Once the compressed instance is gathered onto a single machine, our goal is to rank $T_0$ bottom-up optimally. To this end, we  adapt Lam and Yue's sequential \Labeling{} procedure~\cite{LamY01} to work with our compressed representation. At a retained vertex $v$, given the full critical sets of its children, the \Labeling{} procedure resolves duplicate visible ranks from highest to lowest by repeatedly assigning the smallest unused ranks above the conflicting branches. Here, the challenge is to avoid unpacking star-branch counts into individual branches, as it would cause an $\Omega(N)$ space blowup. It turns out that we can run \Labeling{} directly on the star-branch counts by tracking a single designated star branch per type $\mu$ explicitly, while resolving each group of remaining star branches of type $\mu$ in a single step of \Labeling{}. Since star branches of type $\mu$ only need ranks above $\mu$, they first participate in \Labeling{} once all conflicting ranks above $\mu$ are resolved, at which point all of them except one can receive a contiguous block of free ranks above $\mu$. To stay within the same space bound, we also compress batch-assigned ranks as segments and match them back to the star branches in sorted order at the end. Crucially, as a result of this batched labeling execution, the working space remains $O(N^{3/4}k^{1/4})$ throughout.

\paragraph{Multi-stage ranking across the tree.}
The above space bound of $O(N^{3/4}k^{1/4})$ guides the design of our MPC algorithm. The key observation is that, in a tree of height $H\le D+1$ with $N$ vertices, our threshold size decomposition applied at a subtree $T_w$ with threshold $B$ produces a retained tree $T_0$ with $k\le HN/B$ vertices. If $N^{3/4}k^{1/4}=O(S)$, standard $O(1)$-round MPC primitives suffice to gather the required information and compute a critical ranking of $T_w$ locally. Therefore, we rank the tree bottom-up in stages determined by increasing size thresholds $S=B_1<B_2<\cdots<B_L\ge n$, where stage $j$ computes a critical ranking for every maximal subtree of size at most $B_j$ from the already solved subtrees of size at most $B_{j-1}$. An initial $O(\log D)$-round preprocessing step suffices to root the tree and compute subtree sizes and preorder intervals, so that vertices can identify their stage-$j$ subtrees in $O(1)$ rounds. We distinguish between two cases based on the tree diameter.

When the diameter is small ($D\le n^{2\delta}$), we set $B_{j+1}=S(B_j/H)^{1/4}$, so that every stage-$j$ subtree has compressed size $B_j^{3/4}(HB_j/B_{j-1})^{1/4}=S$ and is solved in $O(1)$ rounds. Under this recurrence, the threshold grows as $H(S/H)^{\frac{4}{3}(1-4^{-j})}$, reaching $n$ in $L=O(\log(1/\delta))=O(1)$ stages and yielding $O(\log D)$ rounds in total. 

When the diameter is large ($D>n^{2\delta}$), a stage-$j$ subtree may have too many retained vertices to fit on one machine at once. We then double the thresholds at each stage ($B_j=2^{j-1}S$). Because a stage-$j$ subtree has at most $2B_{j-1}$ vertices and each retained vertex has subtree size greater than $B_{j-1}$, no retained vertex can have two retained children, so the retained tree $T_0$ is a simple path. We partition this path into $O(B_j^{3/4}H^{1/4}/S)$ contiguous subpaths that each satisfy the $O(S)$ compressed-memory bound and solve them bottom-up in $O(1)$ rounds per subpath, forming a geometric series across stages that sums to $O(n^{3/4}D^{1/4}/S+\log D)$ rounds.
\section{Preliminaries}\label{sec:prelim}
We introduce the terminology and basic lemmas frequently used for tree edge ranking in \cref{sec:basics}. We then introduce preliminaries specific to our results in \cref{sec:com-prel} and \cref{sec:models}.

\subsection{Tree edge ranking basics}\label{sec:basics}
The $n$-vertex input tree $T=(V,E)$ has \emph{diameter} $D$ when its longest simple path has $D$ edges. For a rooted tree $T$, let $T_v$ denote the subtree rooted at $v$ consisting of $v$ and its descendants, and let $|T_v|$ be its size. The \emph{depth} of a vertex is the number of vertices on its path from the root, and the \emph{height} $H$ of the tree is the largest depth; we have $H-1\le D\le 2(H-1)$.

A (valid) \emph{edge ranking} of $T$ is a map $\rho\colon E\to\mathbb Z_{>0}$ that assigns positive integer ranks to edges such that the path between any two distinct edges $e,f$ with $\rho(e)=\rho(f)$ contains an edge $g$ with $\rho(g)>\rho(e)$. An \emph{optimal} edge ranking achieves the least possible maximum rank, denoted by $\er(T)$.

We introduce the concepts of visibility, weight, and criticality. For an edge ranking of a rooted tree $T$, we say that a rank $r$ is \emph{visible} from the root if some edge with rank $r$ has no edge of larger rank on its path to the root (the validity condition ensures that at most one such edge exists). The \emph{visible set} of an edge ranking is the set of ranks that are visible from the root. The \textit{weight} for a set of ranks $A$ is
\[
 \wt(A)=\sum_{a\in A}2^{a-1}.
\]
The \emph{critical set} $\crit(T)$ is the unique visible set of \textit{minimum} weight over all edge rankings of $T$. A ranking whose visible set is critical is a critical ranking. Every critical ranking is optimal (\cref{prop:opt}), but an optimal ranking need not be critical because it may expose extra ranks.

\begin{proposition} \label{prop:opt} For a rooted tree $T$, $\er({T})=\max\crit(T)\le|T|-1$, regardless of the chosen root.
\end{proposition}
\begin{proof}
    For the equality, note that every set containing rank $r$ has weight at least $2^{r-1}$, whereas a set of smaller ranks has weight at most $2^{r-1}-1$. Thus a minimum-weight visible set also minimizes the maximum rank, giving $\er(T)=\max\crit(T)$. For the inequality, assign edges distinct ranks.
\end{proof}

\paragraph{Branches and admissible ranks.}
For a vertex $v$ with children $v_1,\ldots,v_d$, we call the edge $vv_i$ together with the subtree $T_{v_i}$ \emph{branch} $i$. The rank assigned to $vv_i$ is the rank of branch $i$.

Fix an edge ranking of each subtree $T_{v_i}$, with visible set $A_i$. A rank $c$ is \emph{admissible} for branch $i$ when $c\notin A_i$. The branch with rank $c$ then contributes $c$ and its visible ranks above $c$ to the visible set at $v$:
    \begin{equation}\label{eq:contrib}
     \ctr(A_i,c)=\{c\}\cup\{x\in A_i:x>c\}.
    \end{equation}
    
We now establish when admissible ranks give a valid ranking.
\begin{proposition}\label{prop:valid}
Let $v$ have children $v_1,\ldots,v_d$, and fix valid edge rankings of the subtrees $T_{v_i}$ with visible sets $A_i$. Assigning ranks $c_1,\ldots,c_d$ to the branches gives a valid edge ranking of $T_v$ \textit{if and only if} every $c_i$ is admissible for branch $i$ and the sets $\ctr(A_i,c_i)$ are pairwise disjoint.
\end{proposition}
\begin{proof}
    An edge inside $T_{v_i}$ with rank $c_i$ conflicts with $vv_i$  when it is visible from $v_i$, so we need $c_i\notin A_i$. Under this condition, branch $i$ has visible set $\ctr(A_i,c_i)$ at $v$. Two equal-rank edges in different branches conflict when both are visible from $v$, since the path between them passes through $v$. Hence the ranking is valid if and only if, in addition, the sets $\ctr(A_i,c_i)$ are pairwise disjoint.
\end{proof}
    
An assignment of admissible ranks $c_1,\ldots,c_d$ to the branches is \emph{critical} if it gives a ranking of $T_v$ that is valid and whose visible set $\bigcup_{i=1}^d  \ctr(A_i,c_i)$ has minimum weight among all valid assignments. We denote the visible set of a critical assignment by $\Lab(A_1,\ldots,A_d)$. 

The next lemma proves several useful facts and, in particular, that a critical assignment gives a critical ranking, provided that each subtree ranking is critical.
\begin{lemma}\label{lem:dom}
Let $A$ and $C$ be sets of ranks, and let $v$ have children $v_1,\ldots,v_d$. We have the following:
\begin{enumerate}[label=(\alph*),leftmargin=*]
\item If $\wt(A)<\wt(C)$, the rank $c=\max(A\triangle C)$ satisfies $c\in C\setminus A$ and $\ctr(A,c)\subseteq C$.
\item For every positive rank $c\notin A$, we have $\wt(\ctr(A,c))\ge\wt(A)+1$.
\item Critical rankings of the $T_{v_i}$ together with a critical assignment of ranks to the branches form a critical ranking of $T_v$, i.e., $\crit(T_v)=\Lab(\crit(T_{v_1}),\ldots,\crit(T_{v_d}))$.
\end{enumerate}
\end{lemma}
\begin{proof}
(a)~At the largest rank $c$ where $A$ and $C$ differ, the set with larger
weight must include $c$. Since the sets agree above $c$, we have $\ctr(A,c)\subseteq C$.

(b)~Let $c\notin A$. Then $A$ is the disjoint union of $\{x\in A:x<c\}$ and $\{x\in A:x>c\}$, and $\ctr(A,c)$ is the disjoint union of $\{c\}$ and $\{x\in A:x>c\}$, so
\[
 \wt(\ctr(A,c))-\wt(A)=2^{c-1}-\wt(\{x\in A:x<c\}) \ge 1.
\]

(c)~By definition, $\wt(\crit(T_v))\le\wt(\Lab(\crit(T_{v_1}),\ldots,\crit(T_{v_d})))$. Conversely, consider any valid edge ranking of $T_v$ where child $i$ has visible set $A_i'$ and edge $vv_i$ has rank $c_i'$. By \cref{prop:valid}, the sets $C_i = \ctr(A_i',c_i')$ are pairwise disjoint, and the visible set of $T_v$ is $V = \bigcup_{i=1}^d C_i$. Minimality and (b) give $\wt(\crit(T_{v_i}))\le\wt(A_i')<\wt(C_i)$, so (a) yields an admissible rank $c_i\in C_i\setminus\crit(T_{v_i})$ satisfying $\ctr(\crit(T_{v_i}),c_i)\subseteq C_i$. The sets $\ctr(\crit(T_{v_i}),c_i)\subseteq C_i$ are pairwise disjoint, so by \cref{prop:valid} they form a valid ranking of $T_v$ with visible set $V_{\mathrm{new}}\subseteq V$ and $\wt(V_{\mathrm{new}})\le\wt(V)$. By definition of $\Lab$, \[\wt(\Lab(\crit(T_{v_1}),\ldots,\crit(T_{v_d})))\le\wt(V_{\mathrm{new}})\le\wt(V).\] Minimizing over all valid rankings of $T_v$ gives $\wt(\Lab(\crit(T_{v_1}),\ldots,\crit(T_{v_d})))\le\wt(\crit(T_v))$, forcing equality of weights and hence, as $\wt$ is injective, of the sets.
\end{proof}

The above lemma shows that a branch with child weight $\wvar{a}_i$ requires weight $\wvar{a}_i+1$: every set of ranks of at least this weight contains the visible set of a valid branch ranking. In particular, the child weights alone determine $\Lab$:
\[
 \wt(\Lab(A_1,\ldots,A_d))=\min\sum_{i=1}^{d}\wvar{x}_i
\]
where the minimum is over positive integers $\wvar{x}_i\ge\wt(A_i)+1$ whose binary representations have pairwise disjoint $1$-bit positions (e.g., $\wvar{x}_1 = 10 = 1010_2$ and $\wvar{x}_2 = 5 = 0101_2$ for $\wt(A_1)=9$ and $\wt(A_2)=4$). The resulting minimum-weight visible set can be read off from the $1$-bits of the $\wvar{x}_i$'s.

\subsection{Complexity result preliminaries}\label{sec:com-prel}

\paragraph{Boolean circuits.}
A Boolean circuit is a finite directed acyclic multigraph whose vertices are
inputs or gates. Inputs have indegree zero. Each gate $v$ has an ordered
list of $k\in\{0,1,2\}$ incoming wires and a function
$f_v\colon\{0,1\}^k\to\{0,1\}$; a gate with $k=0$ is a constant.
The \emph{fan-in} and \emph{fan-out} of a vertex are its indegree and
outdegree, counting parallel wires separately. In particular, a binary
gate has distinct left and right input slots even if both wires have the
same source. An input assignment determines the value $c_v$ of each
vertex in topological order: inputs take their assigned bits, and a gate
with predecessors $u_1,\ldots,u_k$ takes
$c_v=f_v(c_{u_1},\ldots,c_{u_k})$.
An ordered list of designated sinks specifies the circuit's outputs.
Its \emph{size} is its number of vertices, and its \emph{depth} is the
maximum number of edges on a directed path. A circuit family
$(C_n)_{n\ge1}$ is \emph{logspace-uniform} if a deterministic transducer
outputs $C_n$ from $1^n$ using $O(\log n)$ workspace~\cite{GreenlawHR95}.

\paragraph{Complexity classes.}
For $k\ge1$, the class $\NC^k$ consists of decision problems computed
by uniform, bounded-fan-in Boolean circuits of polynomial size
and depth $O(\log^k n)$. We write $\NC=\bigcup_{k\ge1}\NC^k$.
The class $\Pc$ consists of decision problems solvable in deterministic
polynomial time, and $\NC\subseteq\Pc$~\cite{GreenlawHR95,AroraB09}.

An \emph{$\NC^1$ many-one reduction} from $A$ to $B$ computes a
polynomial-length encoding $f(x)$ by a uniform family of
bounded-fan-in circuits of polynomial size and depth $O(\log |x|)$,
with $x\in A$ if and only if $f(x)\in B$.
These reductions compose. We use them
throughout: a problem is $\Pc$-complete if it belongs to $\Pc$ and every
problem in $\Pc$ reduces to it in this sense~\cite{GreenlawHR95}.

\paragraph{PRAM.}\label{par:pram}
A parallel random-access machine (PRAM) consists of synchronous
processors sharing a random-access memory. Unless specified otherwise,
we use $O(\log n)$-bit words and the \emph{common concurrent-read,
concurrent-write} (common CRCW) model. Multiple processors may read
the same memory cell simultaneously. They may also write to the same
cell simultaneously, provided that they all write the same value.
Polylogarithmic PRAM time with polynomially many processors characterizes
$\NC$, although this equivalence need not preserve the exponent
$k$~\cite{GreenlawHR95}.
Prefix sums of $n$ words take $O(\log n)$ PRAM
time~\cite{LadnerF80}. Binary addition, subtraction, comparison, and
prefix sums are also in $\NC^1$~\cite[Section~3]{Cook85}.

We also use two standard path primitives on directed acyclic graphs
with $n^{O(1)}$ vertices. Reachability takes $O(\log n)$ CRCW time;
minimum-cost paths with nonnegative integer costs of $n^{O(1)}$ bits
take $O(\log^2 n)$ time. Both return path witnesses and use work
polynomial in the graph size and cost bit length~\cite{KarpR90}.

The underlying operation combines pairs of known paths through every
possible intermediate vertex. Each round doubles the path length
covered, so $O(\log n)$ rounds suffice. This is called \emph{matrix
squaring}: it uses Boolean operations for reachability and minimum
and addition for path costs. A Boolean round takes constant CRCW
time. A cost round takes $O(\log n)$ time for additions and comparisons
on costs stored across multiple words; pairwise comparisons select
minima in parallel. Keeping an outgoing edge that reaches the target,
or attains the minimum cost, gives the witnesses. Acyclicity ensures
these choices form paths.

\paragraph{Circuit value.}
The \emph{circuit value problem} (CVP) asks whether a designated output of
a given circuit evaluates to $1$ on a given input assignment. It is
$\Pc$-complete~\cite{Ladner75}. We use the
\emph{Alternating Monotone Fanin 2, Fanout 2 CVP} (\AMtwoCVP{}),
with binary AND and OR gates and fan-out at most two.
Every wire between gates joins opposite types, inputs feed only OR gates,
and the designated output is an OR sink. This restriction remains
$\Pc$-complete under $\NC^1$ reductions with the vertices given in
topological order~\cite[Theorems~6.2.1 and~6.2.3]{GreenlawHR95}.
We number them $1,\ldots,m$ and place the designated sink last; moving
it to the end of a given topological order is an $\NC^1$ relabeling.

\subsection{Algorithmic result preliminaries}\label{sec:models}
A common ingredient to our algorithms is a size-threshold decomposition that allows us to identify small subproblems and the bulk of the tree that joins them. For a threshold $\tau\ge1$, let $T_0$ be the tree induced by the vertices $v$ whose subtrees satisfy $|T_v|\ge\tau$. We call $T_0$ the \emph{retained tree} and its vertices \emph{retained}. The components outside $T_0$ are the \emph{attached subtrees}: the subtrees $T_u$ with $|T_u|<\tau$ whose parent lies in $T_0$. Their rankings can be computed separately and then joined through $T_0$.

\begin{lemma}[Size-threshold decomposition]\label{lem:size-threshold}
Let $T$ have $n\ge \tau$ vertices and height at most $H$. The retained tree $T_0 \subseteq T$ contains the root of $T$, has at most $n/\tau$ leaves and at most $Hn/\tau$ vertices. Each attached subtree of $T$ has fewer than $\tau$ vertices and is connected to $T_0$ by exactly one edge.
\end{lemma}
\begin{proof}
Subtree sizes increase toward the root, so retained vertices form a
connected tree containing the root. Each component outside is a subtree $T_u$ with $|T_u|<\tau$, attached by the edge from $u$ to its parent. For distinct leaves $v$ of $T_0$, the subtrees $T_v$ are disjoint and each has at least $\tau$ vertices, so there are at most $n/\tau$ such leaves. Retained vertices of equal depth have disjoint subtrees, so each of the at most $H$ depth levels contains at most $n/\tau$ retained vertices.
\end{proof}

\paragraph{MPC Primitives.} In MPC, basic computing tasks such as sorting, prefix sums, and aggregation by key take $O(1)$ deterministic MPC rounds when $S=n^{\Omega(1)}$~\cite{Goodrich99,GoodrichSZ11}.
A useful primitive will be the following grouped selection.

\begin{corollary}\label{lem:select}
Let $S\ge n^{3/4}$. Consider key-value pairs of total size $O(n)$, where each value occupies $O(\sqrt{n})$ words and values with the same key are totally ordered by a locally computable comparison. For each key, one can compute a minimum value and the number of pairs in $O(1)$ deterministic MPC rounds, using $O(S)$ local memory and $O(n)$ total memory.
\end{corollary}
\begin{proof}[Proof sketch]
Sort the pairs by key and split each group into chunks of $\Theta(S)$ words, each stored on one machine; a group that fits on one machine is finished locally. Replacing every chunk by its minimum shrinks the total size of a group's candidates by a factor $\Omega(S/\sqrt{n})=\Omega(n^{1/4})$, so after $O(1)$ rounds every group fits on one machine. Counts are computed by aggregation.
\end{proof}

We can root the tree and compute depths, subtree sizes, and preorder numbers deterministically in $O(\log D)$ MPC rounds~\cite{BalliuLMOU23,GuptaLMPSSSUV23}. Preorder numbers then allow us to identify subtrees by their intervals.

\begin{lemma}[\cite{GuptaLMPSSSUV23}]\label{lem:preorder}
For a tree $T_u$, $O(\log D)$ MPC rounds suffice to compute preorder numbers $\pre(v)$ of all vertices $v \in T_u$ with local memory $S = n^{\Omega(1)}$ and $O(n)$ total memory. Then, $T_w=\{v:\pre(w)\le\pre(v)<\pre(w)+|T_w|\}$ for every $w \in T_u$.
\end{lemma}

\section{P-completeness of tree edge ranking}\label{sec:pcomplete}

Our main result (\cref{thm:pcomplete}) is proved by showing the following theorem. Since \TER{} is known to be in $\Pc$, it suffices to show $\Pc$-hardness:
\begin{theorem}\label{thm:pcreduction}
    There exists an $\NC^1$ many-one reduction from \AMtwoCVP{} to \TER{} for trees of diameter at most six.
\end{theorem}
We split its proof into two parts, introducing \GreedyCover{} as an intermediate reduction.

\begin{definition}[\GreedyCover{}]\label{def:greedycover}
    Given an integer $K\ge1$ and $d \ge 1$ positive integer \emph{requirements} $\wvar{q}_{1}, \ldots, \wvar{q}_{d}$ as input, \GreedyCover{} is a decision problem to determine if there exists a partition $X_{1}, \ldots, X_{d}$ of $[K]$, such that $\sum_{x \in X_j} 2^{x-1} \ge \wvar{q}_{j}$ for each $j = 1, \ldots, d$.
\end{definition}
We can write \GreedyCover{} as a decision variant of the optimization program below, where $\&$ denotes the bitwise AND operator.
\begin{equation}\label{eq:greedy-program}
\begin{aligned}
\text{minimize}\qquad& x_1 + \cdots + x_d \\
\text{s.t.}\qquad & x_i \ge \wvar{q}_i & i=1,\ldots, d\\
& x_i \in \mathbb{Z}_{\ge 0} & i=1,\ldots, d\\
& x_i \;\&\; x_j = 0 & 1 \le i < j \le d
\end{aligned}
\end{equation}
In the language of \cref{def:greedycover}, $x_i$ corresponds to the weighted sum $\sum_{x \in X_i} 2^{x-1}$. Since this program is always feasible for sufficiently large $x_i$'s, it attains the optimum at some positive integer. Let $\FGA(\wvar{q}_{1}, \ldots, \wvar{q}_{d})$ denote the optimum. Then \GreedyCover{} is feasible if and only if $\FGA(\wvar{q}_{1}, \ldots, \wvar{q}_{d}) \le 2^{K}-1$.
\paragraph{Roadmap.} We first reduce \GreedyCover{} to \TER{} by constructing
one branch per requirement (\cref{sec:demands}). We then encode circuit
evaluation as a \GreedyCover{} instance by analyzing its greedy execution
(\cref{sec:am2-to-ga}). Composing the two reductions proves \cref{thm:pcreduction}. Finally, we amplify the additive gap from one to any constant $r$ using a rooted product of trees. This shows that approximating $\er(T)$ within any constant additive error remains $\Pc$-hard, even for trees of constant diameter (\cref{sec:approx-inapprox}).

\subsection{GreedyAllocation to Tree Edge Ranking}\label{sec:greedy}\label{sec:demands}
In this section, we reduce \GreedyCover{} to \TER{}. The key observation is that, at the root of a tree, an edge ranking with ranks in $[K]$ is a \GreedyCover{} allocation: every branch at the root is served by its own set of visible ranks, these sets are pairwise disjoint, and a set of ranks can serve a branch if and only if its weight is at least the weight of the branch's critical set (\cref{prop:branch-requirements}). Thus, branches act as requirements and the ranks $1,\ldots,K$ as the weights $2^0,\ldots,2^{K-1}$. It then suffices to build, for every requirement $\wvar{q}$, a branch gadget $B_{\wvar{q}}$ whose vertices lie within distance three of its root and whose critical set has weight exactly $\wvar{q}$ (\cref{lem:gadget}), and to join the gadgets $B_{\wvar{q}_1},\ldots,B_{\wvar{q}_d}$ at a common root.
\begin{lemma} \label{lem:gc2ter}
    Given an input $\left\langle K; \wvar{q}_1, \ldots, \wvar{q}_{d}\right\rangle$ to \GreedyCover{} as specified in \cref{def:greedycover}, one can construct in $\NC^1$ a tree $T(\wvar{q}_1, \ldots, \wvar{q}_d)$ such that $\er(T(\wvar{q}_1, \ldots, \wvar{q}_d)) \le K$ if and only if \GreedyCover{} is feasible.
    Furthermore, $T(\wvar{q}_1, \ldots, \wvar{q}_d)$ has diameter at most six.
\end{lemma}
Let a tree $T$ have root $v$ and branches $B_1,\ldots,B_d$, each including its
edge incident to $v$ and rooted at $v$. We say that a set of ranks $X$ is
\emph{realizable} for a rooted tree if the tree admits a valid edge ranking
whose visible set is contained in $X$. The following proposition shows the relevance of critical sets to realizability and \GreedyCover{}.

\begin{proposition}\label{prop:branch-requirements}
A set of ranks $X$ is realizable for a branch $B$ if and only if
$\wt(X)\ge\wt(\crit(B))$. Consequently, $X$ is realizable for $T$ if and
only if it contains pairwise disjoint subsets $X_1,\ldots,X_d$ with
$\wt(X_i)\ge\wt(\crit(B_i))$ for each $i = 1, \ldots, d$.
\end{proposition}
\begin{proof}
The entire tree is realizable if and only if each branch is realizable with pairwise disjoint visible subsets. Hence it suffices to prove the single branch case.

The only-if part follows from the definition of critical sets: a visible set contained in $X$ has weight at most $\wt(X)$, and $\crit(B)$ has minimum weight among visible sets.

For the if-part, we show the existence of a ranking whose visible set is a subset of $X$. Fix a critical ranking of $B$ with the critical set $C:=\crit(B)$. If $C \subseteq X$, we are done. Otherwise, note that the rank $b$ assigned to the unique edge incident to the root is the minimum element of $C$. On the other hand, as $\wt(X)\ge\wt(C)$, $c := \max(C \triangle X)$ belongs to $X \setminus C$. As $C \setminus X$ is also nonempty, \[c = \max(C \triangle X) > \max(C \setminus X) \ge \min C = b.\]
Now increase the root edge's rank from $b$ to $c$. The new ranking is valid; as $c \notin C$ and $c > b$, no old edge of rank $c$ was originally visible from the root. The new visible set is exactly $\ctr(C, c)$, being a subset of $X$, since $c = \max(C \triangle X)$.
\end{proof}

Write $\wvar{w}(T)=\wt(\crit(T))$, with
$\wvar{w}(T)=\er(T)=0$ for a one-vertex tree, and set
$\FGA(\varnothing)=0$. The proposition gives
\begin{equation}\label{eq:approx-merge}
 \wvar{w}(T)=\FGA\bigl(\wvar{w}(B_1),\ldots,\wvar{w}(B_d)\bigr).
\end{equation}
The highest nonzero bit of $\wvar{w}(T)$ determines $\er(T)$.

Now we move to the construction of a \emph{branch gadget} $B_{\wvar{q}}$ for an arbitrary positive integer $\wvar{q}$, whose critical set has weight exactly $\wvar{q}$. Write the binary expansion of $\wvar{q}-1$ as
\[\wvar{q}=1+\sum_{j=1}^{t}2^{a_j}.\]
Start from an edge $rr_{\wvar{q}}$, rooted at $r$,
and attach a star graph with $(a_j+1)$ edges at $r_{\wvar{q}}$ by one of its leaves
for each $j$.

\begin{lemma}\label{lem:gadget}For any positive integer $\wvar{q}$,
$\wt(\crit(B_{\wvar{q}}))=\wvar{q}$.
\end{lemma}
\begin{proof}
In any ranking, the star with $(a_j + 1)$ edges has a maximum rank of at least $a_j + 1$, visible from $r_{\wvar{q}}$. Since those maximum ranks are distinct, any visible set at $r_{\wvar{q}}$ is of weight at least \[\sum_j2^{a_j}=\wvar{q}-1.\]
Equality is attained when we assign rank $a_j + 1$ to each star's edge incident to $r_{\wvar{q}}$.

By \cref{lem:dom}(b), adding the edge $rr_{\wvar{q}}$ increases the visible
weight by at least one. Give it the smallest rank absent from the critical set. In terms of the binary expansion of the weight, it turns on the least significant zero bit and turns off all the originally trailing ones, as the assignment shadows all the lower ranks which were originally visible. This is identical to adding $1$ to the weight; therefore,\[\wt(\crit(B_{\wvar{q}}))=(\wvar{q}-1) + 1 = \wvar{q}.\]
\end{proof}

Now we obtain $T(\wvar{q}_1, \ldots, \wvar{q}_d)$ by identifying  the roots of $B_{\wvar{q}_1}, \ldots, B_{\wvar{q}_d}$ (see also \cref{fig:rowtree}).

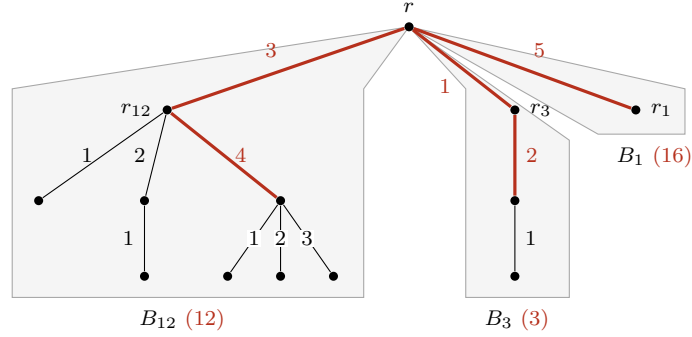
\begin{figure}[htbp]
\centering
\begin{tikzpicture}[
  vtx/.style={circle,fill=black,inner sep=1.2pt},
  lab/.style={font=\scriptsize},
  e/.style={thin},
  ev/.style={very thick,visc},
]
\node[vtx,label={[lab]above:$r$}] (o) at (0,0) {};
\node[vtx,label={[lab]left:$r_{12}$}] (z12) at (-3.2,-1.1) {};
\draw[ev] (o) -- node[lab,above left,text=visc]{$3$} (z12);
\node[vtx] (c1) at (-4.9,-2.3) {};
\draw[e] (z12) -- node[lab,left]{$1$} (c1);
\node[vtx] (c2) at (-3.5,-2.3) {};
\draw[e] (z12) -- node[lab,left]{$2$} (c2);
\node[vtx] (l21) at (-3.5,-3.3) {};
\draw[e] (c2) -- node[lab,left]{$1$} (l21);
\node[vtx] (c4) at (-1.7,-2.3) {};
\draw[ev] (z12) -- node[lab,right,text=visc]{$4$} (c4);
\foreach \x/\l in {-2.4/1,-1.7/2,-1.0/3}
  {\node[vtx] (q) at (\x,-3.3) {}; \draw[e] (c4) -- node[lab,fill=white,inner sep=0.5pt]{$\l$} (q);}
\node[vtx,label={[lab]right:$r_{3}$}] (z3) at (1.4,-1.1) {};
\draw[ev] (o) -- node[lab,below left,text=visc]{$1$} (z3);
\node[vtx] (c32) at (1.4,-2.3) {};
\draw[ev] (z3) -- node[lab,right,text=visc]{$2$} (c32);
\node[vtx] (l32) at (1.4,-3.3) {};
\draw[e] (c32) -- node[lab,right]{$1$} (l32);
\node[vtx,label={[lab]right:$r_{1}$}] (z1) at (3.0,-1.1) {};
\draw[ev] (o) -- node[lab,above right,text=visc]{$5$} (z1);
\begin{scope}[on background layer]
\filldraw[fill=black!4,draw=black!35]
  (o.center) -- (-5.25,-0.82) -- (-5.25,-3.58)
  -- (-0.60,-3.58) -- (-0.60,-0.82) -- cycle;
\filldraw[fill=black!4,draw=black!35]
  (o.center) -- (0.75,-0.82) -- (0.75,-3.58)
  -- (2.12,-3.58) -- (2.12,-1.50) -- cycle;
\filldraw[fill=black!4,draw=black!35]
  (o.center) -- (2.50,-1.42)
  -- (3.65,-1.42) -- (3.65,-0.82) -- cycle;
\end{scope}
\node[lab] at (-3.0,-3.88) {$B_{12}$ \textcolor{visc}{$(12)$}};
\node[lab] at (1.43,-3.88) {$B_3$ \textcolor{visc}{$(3)$}};
\node[lab] at (3.25,-1.72) {$B_1$ \textcolor{visc}{$(16)$}};
\end{tikzpicture}
\caption{The tree $T(12, 3, 1)$ with a valid edge ranking. Red edges are
visible from the root. Numbers in red parentheses give each branch's visible weight, which is at least the corresponding requirement.}
\label{fig:rowtree}
\end{figure}

\begin{proof}[Proof of \cref{lem:gc2ter}] Simply denote $T := T(\wvar{q}_1, \ldots, \wvar{q}_d)$. For the equivalence, \cref{prop:branch-requirements,lem:gadget} give that $[K]$ is realizable for $T$ if and only if \GreedyCover{} is feasible. As the maximum rank is always visible, $\er(T) \le K$ if and only if $[K]$ is realizable.

Now observe the constructed tree. If a requirement $\wvar{q}$ has $\ell$ bits, $B_{\wvar{q}}$ has at most $1+\ell(\ell+1)/2$ vertices other than $r$. Thus $|V(T)|$ is polynomial in the input length, and every vertex is within distance
three of $r$. This shows that the diameter is at most six.

For the parallel construction, compute all binary expansions
$\wvar{q}_i-1$ and the numbers of new vertices in their stars and
branches independently.
Number the common root first, then branches in input order, each with
$r_{\wvar{q}_i}$ first, stars in increasing size, and each center before its
leaves. Prefix sums of these sizes assign consecutive vertex indices,
after which every nonroot vertex independently emits its edge to its
parent. The parallel primitives in \cref{sec:com-prel} give the claimed
$\NC^1$ construction.
\end{proof}

\subsection{AM2CVP to GreedyAllocation}\label{sec:am2-to-ga}
We encode circuit evaluation as a \GreedyCover{} instance by proving the following lemma.

\begin{lemma}\label{lem:am2-to-ga}
Given an $m$-node instance $C$ of \AMtwoCVP{}, one can construct in $\NC^1$
an instance $\mathcal{G}(C)=\langle K;\wvar{q}_1,\ldots,\wvar{q}_d\rangle$ of
\GreedyCover{} that is feasible if and only if $c_m=1$. The construction
has $d=O(m)$ positive requirements of $O(m)$ bits each and $K=18m$.
\end{lemma}

\paragraph{The greedy procedure.}
We first translate \GreedyCover{} into an equivalent run of a \emph{greedy procedure}, following~\cite{delaTorreGS95,OnakP06}. The requirements $\wvar{q}_i$ are fixed input integers. During the run, requirement $i$ has a \emph{residual} $\wvar{r}_i$, initially
$\wvar{q}_i$, recording its remaining requirement. A requirement is
\emph{covered} when its residual becomes zero and is then removed.
The \emph{clocks} $e = K-1,K-2,\ldots,0$ determine which residual is allocated the \emph{weight} $2^e$. \cref{fig:greedy} illustrates two runs of the greedy procedure, defined below.

\begin{definition}[Greedy procedure]\label{def:greedy} Given a positive integer $K$ and a list of requirements, the \emph{Greedy procedure} 
processes \emph{clocks} $e=K-1,\ldots,0$. At clock $e$,
\begin{itemize}[font=\normalfont,nosep]
\item \textup{\textsc{Reject}}. If some residual is at least $2^{e+1}$,
  or two residuals are at least $2^e$, reject.
\item \textup{\textsc{Subtract}}. Otherwise, if exactly one residual is
  at least $2^e$, subtract $2^e$ from it.
\item \textup{\textsc{Cover}}. Otherwise, if a positive residual remains,
  give $2^e$ to a requirement with the largest residual and cover it. Break ties arbitrarily.
\end{itemize}
If no positive residual remains, allocate the remaining weights arbitrarily.
\end{definition}

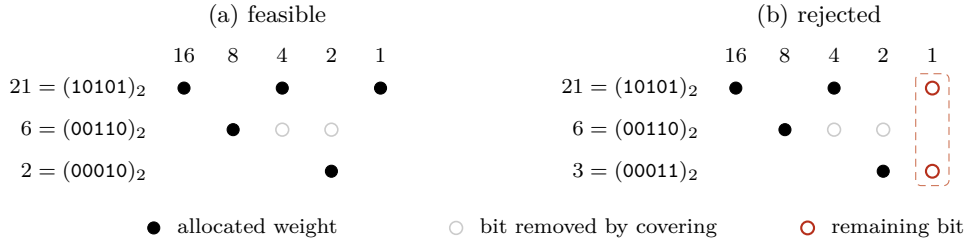
\begin{figure}[htbp]
\centering
\begin{tikzpicture}[
  lab/.style={font=\scriptsize},
  allocated/.style={circle,fill=black,draw=black,inner sep=1.6pt},
  covered/.style={circle,draw=black!45,draw opacity=.45,fill=white,fill opacity=1,line width=.6pt,inner sep=1.7pt},
  remaining/.style={circle,draw=visc,fill=white,line width=.8pt,inner sep=1.7pt},
]
\foreach \off/\title in {0/{(a) feasible},7.3/{(b) rejected}} {
  \node[font=\footnotesize] at (\off+3.5,.95) {\title};
  \foreach \col/\weight in {0/16,1/8,2/4,3/2,4/1}
    \node[lab] at (\off+2.4+.65*\col,.43) {$\weight$};
  \node[lab,anchor=east] at (\off+2.05,0) {$21=(\mathtt{10101})_2$};
  \node[lab,anchor=east] at (\off+2.05,-.55) {$6=(\mathtt{00110})_2$};
  \foreach \col in {0,2}
    \node[allocated] at (\off+2.4+.65*\col,0) {};
  \node[allocated] at (\off+3.05,-.55) {};
  \foreach \col in {2,3}
    \node[covered] at (\off+2.4+.65*\col,-.55) {};
  \node[allocated] at (\off+4.35,-1.10) {};
}
\node[lab,anchor=east] at (2.05,-1.10) {$2=(\mathtt{00010})_2$};
\node[allocated] at (5,0) {};
\node[lab,anchor=east] at (9.35,-1.10) {$3=(\mathtt{00011})_2$};
\node[remaining] at (12.3,0) {};
\node[remaining] at (12.3,-1.10) {};
\draw[visc!65,densely dashed,rounded corners=2pt]
  (12.08,.19) rectangle (12.52,-1.29);
\node[allocated] at (2.0,-1.85) {};
\node[lab,anchor=west] at (2.18,-1.85) {allocated weight};
\node[covered] at (6.0,-1.85) {};
\node[lab,anchor=west] at (6.18,-1.85) {bit removed by covering};
\node[remaining] at (10.65,-1.85) {};
\node[lab,anchor=west] at (10.83,-1.85) {remaining bit};
\end{tikzpicture}
\caption{Binary requirements and their allocations. In (b), two remaining
bits need the final weight $1$.}
\label{fig:greedy}
\end{figure}

\begin{lemma} \label{lem:greedy}
A \GreedyCover{} instance is feasible if and only if the corresponding greedy procedure never rejects.
\end{lemma}

\begin{proof}
We induct on the number $K\ge0$ of
remaining ranks. For $K=0$, feasibility
means that every residual is zero. For $K\ge1$, process the clock $e=K-1$. All weights below $2^e$ sum to $2^e-1$.

A residual of at least $2^{e+1}$ exceeds the total available weight, and two residuals of at
least $2^e$ would both need $2^e$, so rejection is correct.

Otherwise let $k$ have the largest residual $\wvar{r}_k$. If $\wvar{r}_k\ge2^e$, every feasible allocation gives
$2^e$ to $k$. Removing that weight leaves precisely the instance
with residual $\wvar{r}_k-2^e$.
If $\wvar{r}_k<2^e$, modify any feasible allocation to give
$2^e$ alone to $k$. If $2^e$ was unused or already belonged
to $k$, simply free the smaller ranks formerly assigned to $k$.
If it belonged to another requirement $j$, take it from $j$ in exchange for the ranks formerly
assigned to $k$. Since $\wvar{r}_k\ge\wvar{r}_j$, the exchanged ranks cover $\wvar{r}_j$, which preserves feasibility. In both cases, the remaining procedure never rejects if and only if the processed instance is a feasible \GreedyCover{} instance, as guaranteed by the induction hypothesis. Adding back $2^e$ and $k$ preserves the equivalence.
\end{proof}

Recall that an \AMtwoCVP{} instance consists of Boolean input values, nodes representing inputs, AND gates, or OR gates, and directed wires between nodes.

To represent Boolean values, for both inputs and intermediate evaluation results, we use clock-dependent lists of residuals called \emph{signals}. In place of a wire, we implement a \emph{selector}, which converts signals into \emph{messages}, multisets of requirements in a different format for use by subsequent gates. Finally, a gate gadget consumes messages and reproduces a signal corresponding to the evaluation of the node. As the clock decreases, the circuit finally leaves a signal after the last gate evaluation. It is tested against the last \emph{test requirement}, and the procedure rejects if the signal implies that the final evaluation of the circuit is $0$.

\paragraph{Signals.} For each clock $e\ge10$, the signals
\[
 \sigH_e=\{2^e\},\qquad
 \sigL_e=\{2^{e-9}+2^{e-10}+1,\,2^{e-9}+2^{e-10}\}
\]
represent $1$ and $0$, respectively.

\paragraph{Selectors.}
A \emph{selector} uses a signal to choose which residual list, called a \emph{message}, survives after its evaluation. It consists of twelve requirements: two $T$-requirements for its functionality, six
$U$-requirements each carrying a $\wvar{u}$-message element, and four
$V$-requirements each carrying a $\wvar{v}$-message element.
Fix $e\ge16$ and messages
$0\le\wvar{u}_j,\wvar{v}_k<2^{e-14}$ for $j\in[6]$, $k\in[4]$. Define
\begin{equation}\label{eq:selector}
\begin{aligned}
 T_1&=2^{e-1}+2^{e-14},
 &T_2&=2^{e-2}+2^{e-14},\\
 U_j&=2^{e-j-2}+\wvar{u}_j
 &&(j\in[6]),\\
 V_k&=2^{e-k-10}+\wvar{v}_k
 &&(k\in[4]).
\end{aligned}
\end{equation}
Depending on the given signal, a one-clock shift determines which message survives during clocks $e, \ldots, e-14$: $\sigH_e$ consumes clock $e$ so that the next eight clocks \emph{subtract} the leading powers of $T_1,T_2,U_1,\ldots,U_6$ (exposing the $\wvar{u}_j$'s while the $V_k$'s are covered), whereas $\sigL_e$ leaves clock $e$ free so that those eight clocks instead \emph{cover} $T_1,T_2,U_1,\ldots,U_6$ one step early (allowing the last four clocks to subtract the leading powers of $V_1,\ldots,V_4$ and expose the $\wvar{v}_k$'s).

\begin{samepage}
\begin{lemma}[Selector]\label{lem:selector}
At clock $e \ge 16$, suppose the residuals consist of the requirements in
\cref{eq:selector}, a signal $\sigH_e$ or $\sigL_e$, and other residuals
below $2^{e-14}$. Processing clocks $e,e-1,\ldots,e-14$ causes no
rejection and leaves all other residuals unchanged. Among the signal
and selector requirements, the remaining residuals are exactly the
positive entries of:
\begin{itemize}
    \item $\wvar{u}_1,\ldots,\wvar{u}_6$ on $\sigH_e$;
    \item $\wvar{v}_1,\ldots,\wvar{v}_4$ on $\sigL_e$.
\end{itemize}
\end{lemma}
\end{samepage}
\begin{proof}
On $\sigH_e$, the first clock consumes the signal and the next eight
subtract the leading powers of $T_1,T_2,U_1,\ldots,U_6$.
The next three clocks cover $V_1,V_2,V_3$, leaving
$V_4=2^{e-14}+\wvar{v}_4$ and two $T$-residuals equal to $2^{e-14}$.
The last three clocks cover $V_4$ (first if $\wvar{v}_4>0$, or tied with the $T$-residuals if $\wvar{v}_4=0$) and remove both $T$-residuals, leaving all $\wvar{u}_j$'s.

On $\sigL_e$, the first eight clocks instead \emph{cover}
$T_1,T_2,U_1,\ldots,U_6$ one step earlier. The following three weights consume the signal completely, and the last four weights subtract the $V$ leading powers, leaving all $\wvar{v}_k$'s.

The $\wvar{u}$- and $\wvar{v}$-messages, along with all other residuals,
lie below $2^{e-14}$, whereas every
residual that is subtracted from or covered is at least $2^{e-14}$. Thus neither rejection nor unexpected subtraction from other residuals occurs. See \cref{fig:selector} for the visualized procedure.
\end{proof}

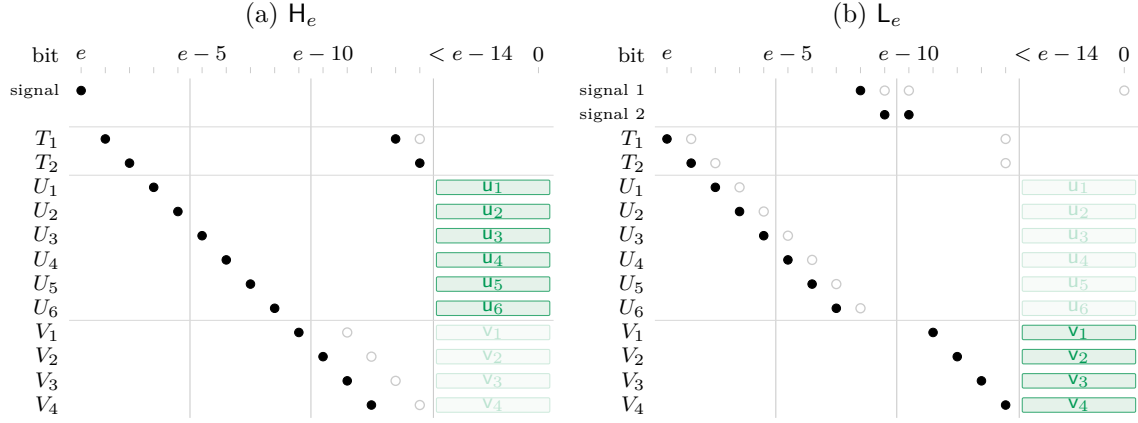
\begin{figure}[htbp]
\centering
\begin{tikzpicture}[
  selallocated/.style={circle,draw=black,fill=black,inner sep=0pt,
    minimum size=3.2pt,line width=0.2pt},
  seldeleted/.style={circle,draw=black!45,draw opacity=.45,
    fill=white,fill opacity=1,inner sep=0pt,minimum size=3.4pt,
    line width=.5pt},
  selrow/.style={font=\scriptsize,anchor=east,inner sep=0pt},
  selhead/.style={font=\scriptsize,inner sep=0pt},
  seltail/.style={draw=delivercol!75,fill=delivercol!9,
    text=delivercol,rounded corners=.5pt,font=\scriptsize,
    inner xsep=2pt,inner ysep=.4pt,minimum width=1.50cm},
]
\newcommand{\selbit}[3]{%
  \node[#3] at ({1+.32*(#1)},{-.32*(#2)}) {};}
\newcommand{\selunit}[1]{%
  \node[seldeleted] at (7.05,{-.32*(#1)}) {};}
\newcommand{\seltailbox}[3]{%
  \node[seltail,opacity=#3] at (6.45,{-.32*(#1)}) {$#2$};}
\newcommand{\selaxes}[1]{%
  \node[font=\small] at (3.65,1.00) {#1};
  \node[selhead,anchor=east] at (.72,.48) {bit};
  \foreach \t/\label in {0/e,5/{e-5},10/{e-10}} {
    \node[selhead] at ({1+.32*\t},.48) {$\label$};
  }
  \node[selhead] at (6.15,.48) {$<e-14$};
  \node[selhead] at (7.05,.48) {$0$};
  \foreach \t in {0,...,14} {
    \draw[black!20,line width=.3pt]
      ({1+.32*\t},.23) -- ({1+.32*\t},.31);
  }
  \draw[black!20,line width=.3pt] (7.05,.23) -- (7.05,.31);
  \foreach \x in {2.44,4.04,5.66} {
    \draw[black!22,line width=.35pt] (\x,.16) -- (\x,-4.33);
  }
  \foreach \y in {-.48,-1.12,-3.04} {
    \draw[black!15,line width=.3pt] (.84,\y) -- (7.25,\y);
  }
  \foreach \j in {1,2} {
    \node[selrow] at (.72,{-.32*(\j+1)}) {$T_{\j}$};
  }
  \foreach \j in {1,...,6} {
    \node[selrow] at (.72,{-.32*(\j+3)}) {$U_{\j}$};
  }
  \foreach \j in {1,...,4} {
    \node[selrow] at (.72,{-.32*(\j+9)}) {$V_{\j}$};
  }}

\begin{scope}
\selaxes{(a) $\sigH_e$}
\node[selrow,font=\tiny] at (.72,0) {signal};
\selbit{0}{0}{selallocated}
\selbit{1}{2}{selallocated}
\selbit{13}{2}{selallocated}
\selbit{14}{2}{seldeleted}
\selbit{2}{3}{selallocated}
\selbit{14}{3}{selallocated}
\foreach \j in {1,...,6} {
  \selbit{\j+2}{\j+3}{selallocated}
  \seltailbox{\j+3}{\wvar{u}_{\j}}{1}
}
\foreach \j in {1,...,4} {
  \selbit{\j+10}{\j+9}{seldeleted}
  \selbit{\j+8}{\j+9}{selallocated}
  \seltailbox{\j+9}{\wvar{v}_{\j}}{.25}
}
\end{scope}

\begin{scope}[xshift=7.75cm]
\selaxes{(b) $\sigL_e$}
\node[selrow,font=\tiny] at (.72,0) {signal 1};
\node[selrow,font=\tiny] at (.72,-.32) {signal 2};
\selbit{9}{0}{seldeleted}
\selbit{10}{0}{seldeleted}
\selbit{8}{0}{selallocated}
\selunit{0}
\selbit{9}{1}{selallocated}
\selbit{10}{1}{selallocated}
\selbit{1}{2}{seldeleted}
\selbit{14}{2}{seldeleted}
\selbit{0}{2}{selallocated}
\selbit{2}{3}{seldeleted}
\selbit{14}{3}{seldeleted}
\selbit{1}{3}{selallocated}
\foreach \j in {1,...,6} {
  \selbit{\j+2}{\j+3}{seldeleted}
  \selbit{\j+1}{\j+3}{selallocated}
  \seltailbox{\j+3}{\wvar{u}_{\j}}{.25}
}
\foreach \j in {1,...,4} {
  \selbit{\j+10}{\j+9}{selallocated}
  \seltailbox{\j+9}{\wvar{v}_{\j}}{1}
}
\end{scope}
\end{tikzpicture}
\caption{Selector allocation; marks are as in \cref{fig:greedy}.
Opaque green boxes show selected tails; faded tails are covered.}
\label{fig:selector}
\end{figure}

\paragraph{Messages.}
Now we define the messages coming from the selectors. For $e\ge13$ and $b\in\{0,1\}$, the left and right input messages at
clock $e$ are
\begin{equation}\label{eq:ports}
\begin{aligned}
 \mathcal M_e^L(b)
 &=\begin{cases}
   \{2^e\},&b=1,\\
   \{2^{e-3}+2^{e-4}\}\cup\sigL_{e-3},&b=0,
  \end{cases}\\[2pt]
 \mathcal M_e^R(b)
 &=\begin{cases}
   \{2^{e-1}+2^{e-3},\,2^{e-2}\},&b=1,\\
   \{2^{e-1}+2^{e-12}+2^{e-13}+1,\,
      2^{e-2}+2^{e-12}+2^{e-13}\},&b=0.
  \end{cases}
\end{aligned}
\end{equation}
The union in \eqref{eq:ports} is a multiset union. Note that $\mathcal{M}_{\ast}^{\ast}(0)$ has at most three requirements, and $\mathcal{M}_{\ast}^{\ast}(1)$ has at most two. We place the bit-$0$ entries in the $\wvar{u}$-slots and the bit-$1$ entries in the $\wvar{v}$-slots, so that the selector of a node whose signal encodes $z$ delivers messages for $1-z$. Since a node has at most two outgoing wires, there are enough message slots to carry them all. Complementing the transmitted bit, although a bit tricky, avoids simulating both AND and OR gates. For each node $i$, let $c_i$ denote its actual evaluation in the circuit. Define:
\[
 z_i:=\begin{cases}
 c_i,&i\text{ is an input or an AND gate},\\
 1-c_i,&i\text{ is an OR gate}.
 \end{cases}
\]
At an AND gate $i$, its OR predecessors $j,k$ satisfy
$1-z_j=c_j$ and $1-z_k=c_k$. At an OR gate, its input or AND
predecessors satisfy $z_j=c_j$ and $z_k=c_k$. Using De Morgan's law, every gate satisfies
\[
 z_i=(1-z_j)\wedge(1-z_k).
\]
Therefore, we evaluate $z_i$ instead, simulating only AND gates and using selectors to carry flipped bits. The messages are therefore designed for AND evaluation.

\paragraph{The AND gate.}
Unlike the selector, the AND gate consumes only two input messages and produces the result after three clocks.

\begin{samepage}
\begin{lemma}[AND gate]\label{lem:gate}
Let $e\ge13$ and $x,y\in\{0,1\}$. At clock $e$, suppose the residuals consist of
$\mathcal M_e^L(x)$, $\mathcal M_e^R(y)$, and other residuals below
$2^{e-3}+2^{e-4}$. Processing clocks $e, e-1, e-2$ causes no rejection
and leaves all other residuals unchanged. Among the input messages,
the remaining residuals at clock $e-3$ are exactly the signal $\sigH_{e-3}$ if $x=y=1$, or the signal $\sigL_{e-3}$ otherwise.
\end{lemma}
\end{samepage}
\begin{proof}
If $x=1$, the first weight consumes the left input. The next two subtract
$2^{e-1}$ and $2^{e-2}$ from the right input, leaving
$\sigH_{e-3}$ if $y=1$, or $\sigL_{e-3}$ if $y=0$.
If $x=0$, the three weights cover the two right residuals and
$2^{e-3}+2^{e-4}$, leaving $\sigL_{e-3}$.
The two right residuals have leading powers $2^{e-1}$ and $2^{e-2}$,
both above $2^{e-3}+2^{e-4}$, while their tails and all other residuals lie below this bound, so no other residual is subtracted or covered. See \cref{fig:gate-bits} for visualized executions.
\end{proof}

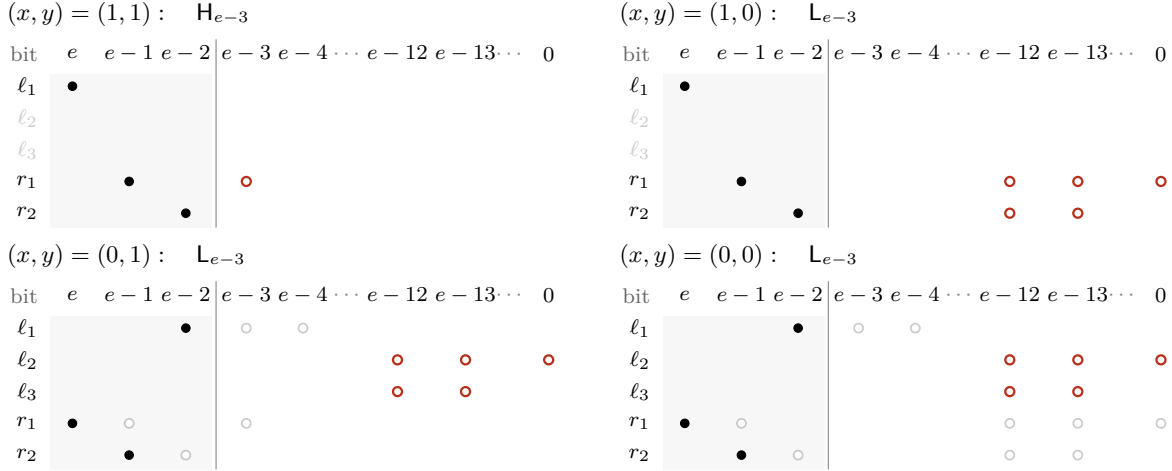
\begin{figure}[htbp]
\centering
\begin{tikzpicture}[
  gbhead/.style={font=\scriptsize,inner sep=0pt},
  gbrowlab/.style={font=\scriptsize,anchor=east,inner sep=0pt},
  gballoc/.style={circle,draw=black,fill=black,inner sep=0pt,minimum size=3.2pt},
  gbremoved/.style={circle,draw=black!45,draw opacity=.45,fill=white,
    fill opacity=1,line width=.65pt,inner sep=0pt,minimum size=3.4pt},
  gbremain/.style={circle,draw=visc,fill=white,fill opacity=1,
    line width=.85pt,inner sep=0pt,minimum size=3.4pt},
]
\newcommand{\gbpanel}[1]{%
  \node[font=\footnotesize,anchor=west] at (0.05,0.96) {#1};
  \fill[black!3] (0.76,0.16) rectangle (2.90,-1.87);
  \draw[black!40,line width=.5pt] (2.95,0.62) -- (2.95,-1.87);
  \node[gbhead,anchor=east,text=black!60] at (.60,.43) {bit};
  \foreach \x/\bit in {1.05/e,1.80/{e-1},2.55/{e-2},3.35/{e-3},4.10/{e-4},5.35/{e-12},6.25/{e-13},7.35/0} {
    \node[gbhead] at (\x,.43) {$\bit$};}
  \foreach \x in {4.70,6.85} {
    \node[gbhead,text=black!45] at (\x,.43) {$\cdots$};}
}
\newcommand{\gbrow}[2]{%
  \node[gbrowlab] at (.60,{-.42*(#1)}) {$#2$};
}
\newcommand{\gbempty}[2]{%
  \node[gbrowlab,text=black!20] at (.60,{-.42*(#1)}) {$#2$};
}
\newcommand{\gbmark}[3]{%
  \node[#3] at (#1,{-.42*(#2)}) {};
}

\begin{scope}
\gbpanel{$(x,y)=(1,1):\quad\sigH_{e-3}$}
\gbrow{0}{\ell_1}\gbempty{1}{\ell_2}\gbempty{2}{\ell_3}
\gbrow{3}{r_1}\gbrow{4}{r_2}
\gbmark{1.05}{0}{gballoc}
\gbmark{1.80}{3}{gballoc}\gbmark{2.55}{4}{gballoc}
\gbmark{3.35}{3}{gbremain}
\end{scope}

\begin{scope}[xshift=8.1cm]
\gbpanel{$(x,y)=(1,0):\quad\sigL_{e-3}$}
\gbrow{0}{\ell_1}\gbempty{1}{\ell_2}\gbempty{2}{\ell_3}
\gbrow{3}{r_1}\gbrow{4}{r_2}
\gbmark{1.05}{0}{gballoc}
\gbmark{1.80}{3}{gballoc}\gbmark{2.55}{4}{gballoc}
\foreach \x in {5.35,6.25,7.35} {\gbmark{\x}{3}{gbremain}}
\foreach \x in {5.35,6.25} {\gbmark{\x}{4}{gbremain}}
\end{scope}

\begin{scope}[yshift=-3.2cm]
\gbpanel{$(x,y)=(0,1):\quad\sigL_{e-3}$}
\gbrow{0}{\ell_1}\gbrow{1}{\ell_2}\gbrow{2}{\ell_3}
\gbrow{3}{r_1}\gbrow{4}{r_2}
\gbmark{2.55}{0}{gballoc}
\foreach \x in {3.35,4.10} {\gbmark{\x}{0}{gbremoved}}
\foreach \x in {5.35,6.25,7.35} {\gbmark{\x}{1}{gbremain}}
\foreach \x in {5.35,6.25} {\gbmark{\x}{2}{gbremain}}
\gbmark{1.05}{3}{gballoc}
\foreach \x in {1.80,3.35} {\gbmark{\x}{3}{gbremoved}}
\gbmark{1.80}{4}{gballoc}
\gbmark{2.55}{4}{gbremoved}
\end{scope}

\begin{scope}[shift={(8.1,-3.2)}]
\gbpanel{$(x,y)=(0,0):\quad\sigL_{e-3}$}
\gbrow{0}{\ell_1}\gbrow{1}{\ell_2}\gbrow{2}{\ell_3}
\gbrow{3}{r_1}\gbrow{4}{r_2}
\gbmark{2.55}{0}{gballoc}
\foreach \x in {3.35,4.10} {\gbmark{\x}{0}{gbremoved}}
\foreach \x in {5.35,6.25,7.35} {\gbmark{\x}{1}{gbremain}}
\foreach \x in {5.35,6.25} {\gbmark{\x}{2}{gbremain}}
\gbmark{1.05}{3}{gballoc}
\foreach \x in {1.80,5.35,6.25,7.35} {\gbmark{\x}{3}{gbremoved}}
\gbmark{1.80}{4}{gballoc}
\foreach \x in {2.55,5.35,6.25} {\gbmark{\x}{4}{gbremoved}}
\end{scope}
\end{tikzpicture}
\caption{Gate inputs after the three shaded weights. Columns are bit
exponents; $\ell_i$ and $r_i$ index the left and right residuals.
Marks are as in \cref{fig:greedy}.}
\label{fig:gate-bits}
\end{figure}

\paragraph{Putting everything together.}\label{sec:map}
Let $C$ be an $m$-node instance of \AMtwoCVP{}, with its input
assignment. Set
\begin{equation}\label{eq:params}
 e_i=18(m-i+1)-1,\qquad K=18m.
\end{equation}
For each nonoutput node $i<m$, list its outgoing wires by destination,
with the left slot before the right. On a wire to node $j$, use
$\mathcal M_{e_j}^L$ or $\mathcal M_{e_j}^R$ according to its input slot:
append the bit-$0$ entries to the $\wvar{u}$-list and the bit-$1$ entries
to the $\wvar{v}$-list, counting parallel wires separately.
Pad the lists with zeroes to lengths six and four, and insert the
selector requirements \eqref{eq:selector} at clock $e_i-3$.

For each input node $i$, also insert the three padding requirements
$2^{e_i},2^{e_i-1},2^{e_i-2}$ and the signal $\sigH_{e_i-3}$ if
its assigned bit is $1$, or $\sigL_{e_i-3}$ if it is $0$.

For the output node $m$, insert the test requirement $2^{e_m-3}=2^{14}$
instead of a selector. These requirements, with multiplicities, define
$\mathcal G(C)=\langle K;\wvar{q}_1,\ldots,\wvar{q}_d\rangle$.

\begin{proof}[Proof of \cref{lem:am2-to-ga}]

The blocks start at $e_1=K-1$ and satisfy $e_{i+1}=e_i-18$.
At node $i$, three clocks process its gate or input padding, and the
next fifteen process its selector if $i<m$.
The selector requirements lie below $2^{e_i-3}$, and every residual
reserved for a later node is at most $2^{e_i-18}$, below the selector
cutoff $2^{e_i-17}$. The test also lies below the output gate's cutoff.
Thus \cref{lem:gate,lem:selector} apply in node order: each node leaves
the signal for $z_i$, then its selector delivers the messages for
$1-z_i$ on its outgoing wires. Zero message entries disappear, so only
the output signal and test remain, with no earlier rejection.
If $z_m=1$, their two requirements of $2^{e_m-3}$ cause rejection.
Otherwise, clocks $e_m-3,e_m-4,e_m-5$ clear the test and low signal.
By \cref{lem:greedy}, feasibility is equivalent to $z_m=0$, hence to
$c_m=1$ because the output is an OR gate.

Each node contributes constantly many requirements, each at most
$2^{K-1}$ and hence of at most $K=18m$ bits. Parallel comparisons
identify the outgoing wires. Every requirement has constantly many
$1$-bits whose positions are determined by node indices and input bits;
prefix sums of their encoded lengths assign output positions.
The primitives in \cref{sec:com-prel} therefore give an $\NC^1$
construction. The source restrictions can also be checked in $\NC^1$.
\end{proof}

\begin{proof}[Proof of \cref{thm:pcreduction}]
Composing \cref{lem:am2-to-ga,lem:gc2ter} gives an $\NC^1$ reduction
to trees of diameter at most six. The tree has $O(m^3)$ vertices,
since there are $O(m)$ requirements of $O(m)$ bits each.
\end{proof}

\subsection{Constant additive gaps}\label{sec:approx-inapprox}
We showed that it is $\Pc$-hard to distinguish $\er(T)\le K$ from $\er(T)\ge K+1$. We now prove that tree edge ranking remains $\Pc$-hard for any constant additive gap $r$, i.e., unless $\NC=\Pc$, no $\NC$ algorithm approximates $\er(T)$ within a constant additive error. Namely, we prove the following extension of \Cref{thm:pcreduction}.

\begin{theorem}[Constant additive gap]\label{thm:approx-gap}
    Fix an integer $r\ge1$. Given a tree $T$ and an integer $K\ge0$,
    promised that either
    \[
     \er(T)\le K \qquad\text{or}\qquad \er(T)\ge K+r,
    \]
    deciding which inequality holds is $\Pc$-complete under
    $\NC^1$ many-one reductions. The hardness remains for trees of diameter at most $6r$.
\end{theorem}

Our approach is to amplify the gap with a rooted product of trees under which edge ranking numbers add up: the $r$-fold product of a hard instance turns a gap of one into a gap of $r$, while the size grows to the $r$-th power and the maximum distance from the root grows by a factor of $r$.

For rooted trees $T$ and $S$,
their \emph{rooted product} $T\odot S$ attaches a separate copy of $S$
to each vertex of $T$ by identifying that vertex with the copy's root.
The product is rooted at the root of $T$; the copy of $S$ attached
there is the \emph{root copy}.

\begin{lemma}\label{lem:approx-product}
For any rooted trees $T,S$, with $h=\er(S)$,
\begin{equation}\label{eq:approx-product}
 \wvar{w}(T\odot S)=2^h\wvar{w}(T)+\wvar{w}(S),
 \qquad
 \er(T\odot S)=\er(T)+\er(S).
\end{equation}
\end{lemma}
\begin{proof} It suffices to show the first identity (since $\wvar{w}(S)<2^h$).
The claim is immediate if either tree is a singleton. Otherwise rank
each copy of $S$ critically
and increase the ranks of a critical ranking of $T$ by $h$.
The original edges separate the copies, giving a valid ranking of
weight $2^{h}\wvar{w}(T)+\wvar{w}(S)$.

For the lower bound, induct on $|V(T)|$. Let $B_1,\ldots,B_d$ be the
branches of $T$, and let $T_i$ be the subtree below the root edge of
$B_i$, so $\wvar{w}(T_i)=\wvar{w}(B_i)-1$.
In $T\odot S$, the corresponding branch joins $T_i\odot S$ to the
root. By induction, its requirement is
\[
 \wvar{w}(T_i\odot S)+1
 =2^h(\wvar{w}(B_i)-1)+\wvar{w}(S)+1.
\]

Fix a feasible allocation at the root of $T\odot S$, and let $x_i$
be the weight assigned to the branch through $T_i\odot S$.
The other allocated bits serve the branches of the root copy of $S$,
so their total weight is at least $\wvar{w}(S)$.
Write $x_i=2^h\wvar{H}_i+\wvar{L}_i$, where
$0\le\wvar{L}_i<2^h$.

If every $\wvar{H}_i\ge\wvar{w}(B_i)$, the high parts form a feasible
allocation for $T$, so $\sum_{i=1}^{d}\wvar{H}_i\ge\wvar{w}(T)$.
Together with the root copy's allocation, this gives total weight
at least $2^h\wvar{w}(T)+\wvar{w}(S)$.

Otherwise, there exists a branch $i$ such that $\wvar{H}_i = \wvar{w}(B_i) - 1$ and $\wvar{L}_{i} \ge \wvar{w}(S) + 1$ (since $\wvar{L}_i<2^h$). Note that there can be at most one such branch; if two distinct branches $i, i'$ are in the same condition, \[
\wvar{L}_{i} + \wvar{L}_{i'} > 2\wvar{w}(S) \ge 2^{h},
\] contradicting that $\sum_{j=1}^{d} \wvar{L}_{j} \le 1 + 2^{1} + \cdots + 2^{h-1} = 2^{h} - 1$. Now, we claim that the root copy of $S$ receives a bit of weight at least $2^h$. Otherwise, its allocated weight $\wvar{L}_S$ is in $[\wvar{w}(S), 2^h - 1]$, and is bitwise disjoint from $\wvar{L}_i$. Thus \[
2^{h} - 1 \ge \wvar{L}_i + \wvar{L}_S \ge 2\wvar{w}(S) \ge 2^{h},
\] leading to a contradiction. Therefore, pick a bit of weight $2^{h'}$, with $h'\ge h$, from the root copy's allocation and reassign it to branch $i$. After dividing the high weights by $2^h$, this gives
\[
\wvar{H}_j' := \begin{cases} \wvar{H}_j & (j \neq i) \\ \wvar{H}_j + 2^{h'-h} & (j = i) \end{cases}.
\]
These allocations remain bitwise disjoint and meet all branch requirements of $T$, so $\sum_{j=1}^{d}\wvar{H}_j'\ge\wvar{w}(T)$. The low part $\wvar{L}_i>\wvar{w}(S)$ is unchanged, and the original total allocated weight is at least
\[
2^h\sum_{j=1}^{d}\wvar{H}_j'+\wvar{L}_i
>2^h\wvar{w}(T)+\wvar{w}(S),
\]
as desired.
\end{proof}

\begin{proof}[Proof of \cref{thm:approx-gap}] The problem is obviously in $\Pc$, so we prove the $\Pc$-hardness.
Define $T^{\odot1}=T$ and
$T^{\odot(j+1)}=T^{\odot j}\odot T$. By the construction and
\cref{lem:approx-product},
\begin{equation}\label{eq:approx-power}
 |V(T^{\odot r})|=|V(T)|^r,
 \qquad \er(T^{\odot r})=r\er(T).
\end{equation}
Given the rooted tree $T$ and threshold $K$ from
\cref{thm:pcreduction}, output $(T^{\odot r},rK)$.
A yes-instance has $\er(T^{\odot r})\le rK$; a no-instance has
$\er(T^{\odot r})\ge r(K+1)=rK+r$. The construction of \cref{thm:pcreduction} places every vertex of $T$ within distance three of its root, so every vertex of $T^{\odot r}$ is within distance $3r$ of its root. Thus the hardness remains for trees of diameter at most $6r$.
\end{proof}
\section{Approximate tree edge ranking in NC}\label{sec:approx}\label{sec:approx-nc}
By \cref{thm:approx-gap}, no $\NC$ algorithm achieves a constant additive error unless $\Pc=\NC$. In this section, we obtain a tree edge ranking with an additive error of $\log_2 n/\log_2\log_2 n$ in $\NC$.

Our approach generalizes the single-separator approach of~\cite{IyerRV91} via our size-threshold decomposition (\cref{lem:size-threshold}): for a parameter $b\ge 2$, we retain a subtree $T_0$ with at most $b$ leaves such that every attached subtree $T_i$ has fewer than $n/b$ vertices. Ranking the attached subtrees recursively and combining them optimally across $T_0$, however, would require their full critical sets, which are $\Pc$-complete to compute in parallel. To avoid this bottleneck, we require the edge connecting each attached subtree to receive a rank strictly above the maximum rank $h_i$ used inside it, allowing us to contract this attached subtree into a single leaf edge incident to $T_0$ with lower bound $h_i+1$.

Our algorithm is presented in three parts. In \cref{sec:approx-replacement}, we develop a lower-bound replacement technique and show that it loses at most one rank per recursion level, while simplifying the problem (\cref{lem:approx-rounding}). In \cref{sec:approx-exact}, we solve the simplified problem, where at most $b$ branches of $T_0$ meet at any vertex and every other branch is a single leaf edge. Combining branches at path endpoints (\cref{lem:approx-branch}) and extending rankings along paths (\cref{lem:approx-path}) then takes $O(b\log^2 n)$ time and $2^{O(b)}n^{O(1)}$ work (\cref{lem:approx-pendant}). Finally, \cref{sec:approx-proof} chooses $b=\lceil\log_2 n\rceil$ to balance polynomial work against $\lfloor\log_b n\rfloor$ recursion levels, proving \cref{thm:approx-nc-main}.

\subsection{Lower-bound replacement}\label{sec:approx-replacement}
Consider a rooted tree $T$ partitioned by \cref{lem:size-threshold} into a retained tree $T_0$ and attached subtrees $T_i$, each equipped with a valid ranking of maximum rank $h_i$ (with $h_i=0$ when $|T_i|=1$). Contracting each $T_i$ into its root turns its edge to $T_0$ into a leaf edge of the contracted tree $T'$ with lower bound $h_i+1$. Because a rank of at least $h_i+1$ on this edge exceeds every rank in $T_i$, it hides $T_i$ from the rest of the tree, so any valid ranking of $T'$ that respects the lower bounds extends directly to a valid ranking of $T$ by restoring the ranked subtrees $T_i$.
The lower bounds may force larger ranks, but they cost at most one rank beyond the additive error of the attached subtrees, while keeping the maximum rank below $|T|$.

\begin{lemma}[Lower-bound replacement]\label{lem:approx-rounding}
Suppose each attached subtree $T_i$ is given a valid ranking of maximum rank $h_i$ satisfying $\er(T_i)\le h_i\le\er(T_i)+E$ for a nonnegative integer $E$. Among valid rankings of $T'$ that respect the lower bounds $h_i+1$, any ranking with the smallest maximum rank extends the given rankings to a valid ranking of $T$ of maximum rank at most $\er(T)+E+1$. Moreover, if $h_i\le |T_i|-1$ for every $i$, this maximum rank is at most $|T|-1$.
\end{lemma}
\begin{proof}
Any valid ranking of $T$ induces a valid ranking of $T'$ by assigning each contracted leaf edge the maximum rank in its branch (the subtree $T_i$ together with its edge to $T_0$): any separator between equal ranks must be an edge of $T_0$, since a contracted branch contains no rank larger than its maximum. Applying this to an optimal ranking of $T$ gives each contracted leaf edge a rank at least $\er(T_i)$; increasing every rank by $E+1$ then respects all lower bounds $h_i+1$ with maximum rank at most $\er(T)+E+1$. Similarly, if $h_i\le |T_i|-1$ for every $i$, applying this to any bijection from $E(T)$ to $[|T|-1]$ assigns each contracted leaf edge the maximum of $|T_i|$ distinct ranks, which is at least $|T_i|\ge h_i+1$, yielding a valid ranking of $T'$ with maximum rank at most $|T|-1$. Restoring each $T_i$ preserves the maximum rank.
\end{proof}

\subsection{Lower-bounded critical ranking}\label{sec:approx-exact}
The replacement produces a rooted tree, again denoted by $T$, formed by a retained tree $T_0$ with at most $b$ leaves together with leaf edges incident to vertices of $T_0$. To solve $T$ bottom-up, we extend criticality to rankings that respect the lower bounds.
\begin{definition}[Lower-bounded critical ranking]\label{def:lower-bound-ranking}
For a tree $T$ with lower bounds on leaf edges, let $\crit(T)$ be the visible set of minimum weight $\wvar{w}(T)=\wt(\crit(T))$ over all valid edge rankings of $T$ that respect the lower bounds. A ranking attaining this minimum is \emph{critical}.
\end{definition}
Minimizing the visible weight also minimizes the maximum rank (\cref{prop:opt}), and since assigning distinct ranks in $[n,2n-1]$ to the edges of $T$ respects all lower bounds, every critical ranking uses ranks below $2n$ and has an $O(n)$-bit weight. Moreover, because edges of $T_0$ carry no lower bounds, the dominance argument of \cref{lem:dom}(c) still applies at every non-root $v\in T_0$: replacing the ranking of $T_v$ by a critical one and adjusting only the edge from $v$ to its parent exposes a subset of the previously visible ranks without violating any bound.

Since $T_0$ has at most $b$ leaves, its root, leaves, and vertices with at least two children in $T_0$ (fewer than $b$ in total) partition $T_0$ into $O(b)$ maximal paths whose internal vertices each have one child in $T_0$. We compute a critical ranking of $T$ bottom-up along $T_0$ by combining branches at path endpoints (\cref{lem:approx-branch}) and extending rankings along paths (\cref{lem:approx-path}).

\paragraph{Combining branches at a path endpoint.}
Consider a vertex $u\in T_0$ whose children in $T_0$ are $v_1,\ldots,v_d$ ($d\le b$) and whose leaf edges have lower bounds $\ell_1\ge\cdots\ge\ell_t$, where the branches $B_i=\{uv_i\}\cup T_{v_i}$ have already been critically ranked. Searching directly over the ranks of $uv_1,\ldots,uv_d$ would take $n^{\Omega(d)}$ work. Instead, we allocate ranks from $2n$ down to $1$: because ranks above $h$ contribute multiples of $2^h$ while the remaining ranks $1,\ldots,h$ sum to $2^h-1$, the weight still needed at rank $h$ by an unfinished branch $B_i$ is uniquely determined by $\wvar{w}(B_i)\bmod 2^h$, so it suffices to track the current rank $h$, the next leaf edge to rank, and the subset of unfinished branches ($2^d$ possibilities).

\begin{lemma}[Combining branches]\label{lem:approx-branch}
Let $u\in T_0$, let $v_1,\ldots,v_d$ ($d\le b\le\lceil\log_2 n\rceil$) be its children in $T_0$, and let each branch $B_i=\{uv_i\}\cup T_{v_i}$ have a given critical ranking. A critical ranking of $T_u$ that agrees with them on $T_{v_1},\ldots,T_{v_d}$ can be computed in $O(\log^2 n)$ time with $2^{O(d)}n^{O(1)}$ work.
\end{lemma}
\begin{proof}
We first reduce the problem to allocating ranks: each branch $B_i$ needs a set of ranks of weight at least $\wvar{w}(B_i)$, each leaf edge needs one rank respecting its bound, and all allocated ranks must be distinct. We then compute a minimum-weight allocation as a shortest path in a directed acyclic graph whose arcs allocate the ranks from $2n$ down to $1$ and whose states record only the current rank, the next leaf edge to rank, and the set of unfinished branches.

\paragraph{Allocating ranks.} By \cref{prop:valid}, every valid ranking of $T_u$ that respects the lower bounds exposes pairwise disjoint sets $X_i$ on the branches $B_i$ and distinct ranks $r_j\ge\ell_j$ outside them on the leaf edges, with $\wt(X_i)\ge\wvar{w}(B_i)$ since its restriction to $B_i$ respects the bounds of $B_i$. Conversely, for any such sets and ranks, \cref{prop:branch-requirements} (whose proof still holds with lower bounds) changes only the rank of each $uv_i$, which carries no bound, so that $B_i$ exposes a subset of $X_i$; this yields a valid ranking of $T_u$ that respects the lower bounds, agrees with the given rankings on $T_{v_1},\ldots,T_{v_d}$, and has visible weight at most $\sum_i\wt(X_i)+\sum_j 2^{r_j-1}$. Hence $\wvar{w}(T_u)$ is the minimum of this total weight, where all ranks may be taken from $[2n]$ since critical rankings use ranks below $2n$. In a minimizing choice, we may assume $r_1>\cdots>r_t$, since swapping two inverted leaf ranks preserves both bounds as $\ell_1\ge\cdots\ge\ell_t$, and, allocating ranks from $2n$ down to $1$, that a branch receives no further rank once its ranks reach weight $\wvar{w}(B_i)$, since discarding such a rank would lower the weight. Before rank $h$ is allocated, every rank already given to an unfinished branch $B_i$ contributes a multiple of $2^h$, while the ranks $1,\ldots,h$ have total weight $2^h-1$; hence, if $B_i$ can still reach $\wvar{w}(B_i)$, it has received weight exactly $\wvar{w}(B_i)-R_i(h)$ and still needs $R_i(h)=\wvar{w}(B_i)\bmod 2^h\in[1,2^h-1]$.

\paragraph{Minimum-weight allocation.} We therefore form a directed acyclic graph on the states $(h,j,\mathcal B)$, where $h\in[0,2n]$, $j\in[1,t+1]$ is the next leaf edge to rank, and $\mathcal B\subseteq[d]$ is the set of unfinished branches, with $R_i(h)>0$ for all $i\in\mathcal B$. For $h\ge1$, the arcs from $(h,j,\mathcal B)$ lead to level $h-1$ and leave rank $h$ unused at cost $0$, assign it to leaf edge $j$ at cost $2^{h-1}$ if $j\le t$ and $h\ge\ell_j$ (advancing $j$ to $j+1$), or assign it to a branch $i\in\mathcal B$ at cost $2^{h-1}$ (removing $i$ from $\mathcal B$ if $2^{h-1}\ge R_i(h)$). We omit the arcs that leave some $i\in\mathcal B$ without rank $h$ while $R_i(h)\ge 2^{h-1}$, since $B_i$ can then no longer reach $\wvar{w}(B_i)$; along every remaining arc, each branch that stays in $\mathcal B$ still needs exactly $R_i(h-1)$. Thus every path from $(2n,1,[d])$ to $(0,t+1,\varnothing)$ describes a feasible choice of the same weight, and conversely a minimizing choice of the above form traces such a path. The shortest-path distance between these two states, in a graph with $O(n(t+1)2^d)=2^{O(d)}n^{O(1)}$ states and $O(n)$-bit costs, is therefore $\wvar{w}(T_u)$. Via PRAM path primitives, we compute a shortest path in $O(\log^2 n)$ time with $2^{O(d)}n^{O(1)}$ work, and \cref{prop:branch-requirements} turns it into a critical ranking of $T_u$.
\end{proof}

\paragraph{Ranking along a path.}
Let $u=v_0,v_1,\ldots,v_m=v$ ($m\ge 1$) be a path in $T_0$ whose internal vertices $v_1,\ldots,v_{m-1}$ each have one child in $T_0$, and let $B=\{uv_1\}\cup T_{v_1}$ be the branch of $T_u$ containing $v$. Given a critical ranking of $T_v$, we rank the $m$ path edges and the leaf edges at $v_1,\ldots,v_{m-1}$. Because $m$ can be $\Omega(n)$, sequential extension along the path would be too slow in $\NC$. Guessing the path edge of largest rank $k$ splits the path into two subpaths ranked below $k$: the lower subpath toward $v$ is hidden by $k$ and needs only a feasible ranking, while the upper subpath toward $u$ determines the visible weight at $u$.

\begin{lemma}[Ranking along a path]\label{lem:approx-path}
Given a critical ranking of $T_v$, a critical ranking of the branch $B$ of $T_u$ containing $v$ that extends it can be computed in $O(\log^2 n)$ time with polynomial work.
\end{lemma}
\begin{proof}
By the dominance argument above, some critical ranking of $B$ extends the given ranking of $T_v$, so we keep it fixed. Edges above $v$ can only conflict with the ranks visible from $v$, namely $\crit(T_v)$. Hence, we may replace $T_v$ by one leaf edge at $v$ of fixed rank $r$ for each $r\in\crit(T_v)$. Every leaf edge on the path (at $v_1,\ldots,v_m$, with none at $v_0=u$) now carries a lower bound or a fixed rank, which we call its \emph{constraint}; all constraints are below $2n$. We solve this path instance by a dynamic program over subpaths and rank limits whose recurrence guesses the path edge of largest rank $k$. We first show that $k$ determines which leaf edges are ranked above $k$, then set up the recurrence, and finally evaluate it in $O(\log^2 n)$ time.

\paragraph{Leaf edges above the largest path-edge rank.}
Fix a subpath with at least one path edge that admits a ranking in $[h]$ together with its leaf edges, and among its minimum-weight such rankings choose one maximizing the largest path-edge rank $k$. Only path edges can separate two equal ranks, and none exceeds $k$, so every rank of at least $k$ occurs at most once. Moreover, a leaf edge with a lower bound is ranked above $k$ if and only if its bound exceeds $k$: if its bound were at most $k$ and its rank $p$ exceeded $k$, swapping its rank with that of the path edge ranked $k$ would keep the ranking valid and respect all bounds without increasing the visible weight, yet raise the largest path-edge rank to $p$, contradicting the choice of $k$. Thus the leaf edges ranked above $k$ are precisely those with constraint in $(k,h]$, and no leaf edge has constraint $k$, since rank $k$ is taken by a path edge.

\paragraph{Dynamic program.}
For $0\le s\le t\le m$ and $0\le h\le 2n$, let $\operatorname{PathDP}([s,t],h)$ be the minimum visible weight at $v_s$ of a valid ranking in $[h]$ of $[s,t]=v_s,\ldots,v_t$ and its leaf edges of constraint at most $h$, or $\infty$ if none exists. An entry is \emph{feasible} if its value is finite. Since critical rankings use ranks below $2n$, $\wvar{w}(B)=\operatorname{PathDP}([0,m],2n)$.

Whenever visible leaf edges must receive distinct ranks from an interval, we reserve their fixed ranks and assign the remaining leaf edges in increasing order of bound, giving each the smallest free rank that respects its bound. An exchange argument shows that these greedy ranks are pointwise smallest, so the greedy assignment fits whenever any assignment does and then minimizes the weight. If $p_j$ is the position, in the increasing list of unreserved ranks, of the first rank respecting the $j$-th smallest bound ($p_j=\infty$ if none exists), then the $j$-th greedy rank has position $a_j=j+\max_{1\le j'\le j}(p_{j'}-j')$, so prefix maxima compute the greedy assignment in $O(\log n)$ time with polynomial work (\hyperref[par:pram]{PRAM primitives}).

For $s=t$, all leaf edges at $v_s$ are visible, so $\operatorname{PathDP}([s,s],h)$ is the weight of their fixed and greedy ranks in $[h]$, or $\infty$ if the greedy assignment does not fit. For $s<t$, a \emph{split} $(i,k)$ ($s\le i<t$, $k\in[h]$) assigns the largest path-edge rank $k$ to $v_iv_{i+1}$. Let $\wvar{c}_k$ be $2^{k-1}$ plus the weight of the fixed and greedy ranks in $(k,h]$ of the leaf edges on $[s,t]$ with constraint in $(k,h]$, or $\infty$ if the greedy assignment does not fit or some leaf edge on $[s,t]$ has constraint $k$. Since edge $v_iv_{i+1}$ separates $[s,i]$ from $[i+1,t]$ and hides $[i+1,t]$ from $v_s$, a split with $\wvar{c}_k<\infty$ combines rankings of $([s,i],k-1)$ and $([i+1,t],k-1)$ into a valid ranking whose visible weight is $\wvar{c}_k$ plus that of the upper part. Conversely, by the characterization above, some optimal ranking of $([s,t],h)$ arises in this way. Hence $([s,t],h)$ is feasible if and only if some split $(i,k)$ has $\wvar{c}_k<\infty$ and both $([s,i],k-1)$ and $([i+1,t],k-1)$ feasible (a \emph{feasible split}), and
\[
  \operatorname{PathDP}([s,t],h)
  =\min_{\substack{s\le i<t,\;1\le k\le h \\ ([i+1,t],k-1)\text{ feasible}}}\bigl(\operatorname{PathDP}([s,i],k-1)+\wvar{c}_k\bigr).
\]

\paragraph{Parallel evaluation.}
Because a split depends on two subpaths, evaluating feasibility by subpath length would take $\Omega(m)$ stages; instead, we use $\lceil\log_2(m+1)\rceil$ doubling rounds. Round $0$ marks the feasible single-vertex entries. Round $r\ge 1$ forms a directed graph on the $O(n^3)$ entries by adding an arc from $([s,t],h)$ to either part of a split $(i,k)$ whenever $\wvar{c}_k<\infty$ and the opposite part was marked in an earlier round, and marks every entry that can reach an earlier-marked entry. Every marked entry is feasible: by induction, the endpoint and the side parts of its path are feasible, so feasibility propagates backward along the path. Conversely, by induction on $r$, every feasible entry on at most $2^r$ vertices is marked after round $r$: repeatedly moving to the larger part of a feasible split while more than $2^{r-1}$ vertices remain leaves aside only subpaths of at most $2^{r-1}$ vertices, which are already marked, tracing a path to a marked entry. Thus $O(\log n)$ rounds of CRCW reachability find all feasible entries in $O(\log^2 n)$ time.

Once feasibility is known, each term of the recurrence depends only on $([s,i],k-1)$, so the recurrence reduces to single-source shortest paths. We add an arc of cost $\wvar{c}_k$ from $([s,i],k-1)$ to $([s,t],h)$ whenever $\wvar{c}_k<\infty$ and $([i+1,t],k-1)$ is feasible, plus an arc from a source to each feasible single-vertex entry of cost equal to its value. Since arcs between entries increase subpath length, the graph is acyclic and distances equal $\operatorname{PathDP}$ values. Via the \hyperref[par:pram]{PRAM path primitives}, we compute all distances and one minimizing split per feasible entry in $O(\log^2 n)$ time and polynomial work. Both parts of a minimizing split are feasible, so following the stored splits from $([0,m],2n)$ into both parts, which reachability does in parallel, ranks every edge; restoring $T_v$ gives a critical ranking of $B$.
\end{proof}

Alternating \cref{lem:approx-branch,lem:approx-path} bottom-up along the $O(b)$ path endpoints and maximal paths of $T_0$ solves the entire tree $T$.
\begin{lemma}[Computing lower-bounded critical rankings]\label{lem:approx-pendant}
Let $T$ have at most $n$ vertices and lower bounds at most $n$, and let $T_0$ have at most $b$ leaves, where $2\le b\le\lceil\log_2 n\rceil$. Then, a lower-bounded critical ranking of $T$ can be computed in $O(b\log^2 n)$ time with $2^{O(b)}n^{O(1)}$ work.
\end{lemma}
\begin{proof}
We first apply \cref{lem:approx-branch} with $d=0$ at each leaf of $T_0$ to rank its incident leaf edges, and then alternate \cref{lem:approx-path} and \cref{lem:approx-branch} upward along $T_0$: \cref{lem:approx-path} extends each solved subtree $T_v$ along the path to its upper endpoint $u$, and \cref{lem:approx-branch} combines the resulting branches and leaf edges at $u$ into a critical ranking of $T_u$. Since $T_0$ has at most $b$ leaves and fewer than $b$ vertices with at least two children in $T_0$ (each with at most $b$ children in $T_0$), there are $O(b)$ paths and path endpoints. Processing them in bottom-up order takes $O(b\log^2 n)$ time and $O(b)\cdot 2^{O(b)}n^{O(1)}=2^{O(b)}n^{O(1)}$ work, and all DP tables and transitions have $O(\log n)$-bit indices for $b\le\lceil\log_2 n\rceil$ and are constructed uniformly in parallel.
\end{proof}

\subsection{Analysis of the recursion}\label{sec:approx-proof}
We now combine the replacement reduction (\cref{lem:approx-rounding}) and the exact solver (\cref{lem:approx-pendant}) with $b=\lceil\log_2 n\rceil$ fixed throughout the recursion. On a rooted subtree $T$ of size $N$, if $N\le b$, we compute an exact critical ranking of $T$ directly via \cref{lem:approx-pendant} (with $T_0=T$ and no lower bounds). If $N>b$, we apply \cref{lem:size-threshold} with threshold $\tau=\lceil N/b\rceil$ to retain the subtree $T_0$ of vertices with subtree size at least $N/b$. By \cref{lem:size-threshold}, $T_0$ has at most $b$ leaves and every attached subtree $T_i$ has fewer than $N/b$ vertices. We rank the attached subtrees $T_i$ recursively in parallel, contract each $T_i$ into a leaf edge of $T'$ with lower bound $h_i+1$ (where $h_i$ is the maximum rank computed for $T_i$), compute a critical ranking of $T'$ via \cref{lem:approx-pendant}, and restore the ranked subtrees $T_i$.

\begin{proof}[Proof of \cref{thm:approx-nc-main}]
Since each recursive call reduces the subtree size by a factor of at least $b$, the recursion has depth at most $\lfloor\log_b n\rfloor$. The base cases are solved exactly, with additive error $0$ and maximum rank at most $|T|-1$. By induction up the recursion using \cref{lem:approx-rounding}, each non-base call increases the additive error by at most $1$ and returns a valid ranking of maximum rank at most $|T|-1<n$, so every lower bound is at most $n$ and each contracted tree $T'$ satisfies the hypotheses of \cref{lem:approx-pendant}. Over $\lfloor\log_b n\rfloor$ levels, the resulting ranking of the $n$-vertex input tree $T$ has maximum rank at most
\begin{equation*}
 \er(T)+\lfloor\log_b n\rfloor
 \le\er(T)+\frac{\log_2 n}{\log_2\log_2 n}.
\end{equation*}

At each recursion level, rooting, computing subtree sizes, and decomposing the active subtrees take $O(\log^2 n)$ time via parallel reachability and prefix sums, and the independent contracted trees on that level are solved in parallel in $O(b\log^2 n)=O(\log^3 n)$ time by \cref{lem:approx-pendant}, for $O(\log^4 n)$ time across all $O(\log n)$ levels. Since the retained trees $T_0$ of all recursive calls are pairwise disjoint, there are at most $n$ subproblems in total, each using $2^{O(b)}n^{O(1)}=n^{O(1)}$ work. Thus, every edge receives its final rank in the unique call in which it belongs to $T_0$ or becomes a contracted leaf edge.
\end{proof}
\section{Massively parallel tree edge ranking}\label{sec:mpc}
Our MPC algorithm for tree edge ranking is presented in four parts. In \cref{sec:mpc-onepermax}, we introduce the compression of visible sets using \textit{segments} and \textit{star branches} to reduce memory consumption.
In \cref{sec:labeling}, we simulate the sequential algorithm of~\cite{LamY01} directly on these compressed representations via \textit{batched labeling}, in polynomial time and without asymptotic space overhead.
In \cref{sec:mpc-instance}, we use our size-threshold decomposition of \cref{lem:size-threshold} to solve subtrees on a single machine in $O(1)$ MPC rounds.
Finally, \cref{sec:mpc-stages} combines these ingredients to establish our main result.

\subsection{Compression of visible sets}\label{sec:mpc-onepermax}
Our construction builds on the compression technique introduced by~\cite{LamY01}. They use this technique to reduce the \textit{computation} required to find an optimal edge ranking. It turns out that the same technique, with some new ingredients, can be used to reduce the space occupied by visible ranks. 

Consider a vertex $v$ whose children carry critical sets $A_1,\ldots,A_d$, so that a critical assignment to its branches gives $\crit(T_v)=\Lab(A_1,\ldots,A_d)$ by \cref{lem:dom}(c). 
Computing an optimal ranking requires the visible sets of the children, but a visible set may contain $\Omega(n)$ ranks; e.g., a star with $n-1$ leaves rooted at its center has $\crit(T)=[n-1]$. We use two ways to represent sets of ranks succinctly, avoiding a potential $\Omega(n)$-blowup. 
\begin{itemize}
    \item \textbf{Segment representation:} A set of ranks $A$ is stored as a list of its maximal intervals of consecutive integers, called \emph{segments}. Segments correspond to maximal runs of $1$-bits in the binary representation of $\wt(A)$; we denote their number by $\nint{A}$, and their representation uses $O(\nint{A})$ words. Weight comparisons and bit operations can be performed on these lists without expanding them.
    \item \textbf{Star branches of type $\mu$:} All branches with visible set $[\mu]$ are stored by their type $\mu$ and their count, using $O(1)$ words. Leaf children, whose set is $\varnothing=[0]$, are stored as star branches of type $0$. Note that each star branch has weight $2^\mu-1$ and must receive a rank greater than $\mu$.
\end{itemize} 
For segment representation, we use the following bound on the number $\nint{A}$ of segments proved in \cite{LamY01}. 
\begin{lemma}[Lemma~5.1 of \cite{LamY01}]\label{lem:seg}
  If $[x_1,y_1],\ldots,[x_q,y_q]$ are the segments of $\crit(T)$, listed increasingly, for an $s$-vertex tree $T$, then $\sum_j x_j<s$ and $\nint{\crit(T)}<\sqrt{s}$.
   Moreover, for pairwise disjoint subtrees $T_1,\ldots,T_f$ of sizes $s_1,\ldots,s_f$ with $\sum_i s_i\le N$,
    \begin{equation}\label{eq:cs}
     \sum_{i=1}^{f}\nint{\crit(T_i)}\le\sum_{i=1}^{f}\sqrt{s_i}\le\sqrt{fN}.
    \end{equation}
\end{lemma}
\begin{proof}[Proof Sketch]
   Consecutive segments are separated by a gap, so $x_j\ge x_{j-1}+2$, and $x_j\ge2j-1$. Summing up, $\sum_j x_j \ge 1+3+\cdots+(2q-1)=q^2$. It only remains to show that $\sum_j x_j < s$ for $\nint{\crit(T)}<\sqrt{s}$ to hold. This follows from the fact that we can associate each $x_j$ with a disjoint subtree that has at least $x_j$ edges by definition of critical ranking. 
    \Cref{eq:cs} then follows from Cauchy--Schwarz.
\end{proof}

For star branch representation, we show that among branches sharing the same maximum rank $\mu$, all but one branch of minimum weight can be replaced by star branches of type $\mu$.
\begin{lemma}[$\mu$-star branch replacement]\label{lem:star_rep}
For any $\mu \ge 1$, let $I\subseteq[d]$ be a set of branches whose critical sets have maximum $\mu$, that is, $\max A_i = \mu$ for all $i \in I$, and let $j \in I$ be a branch of minimum weight in $I$. 
Then, (i) some critical assignment to the branches gives every branch in $I \setminus \{j\}$ a rank strictly greater than $\mu$; and (ii) $\Lab(A_1,\ldots,A_d)$ is unchanged if every branch in $I\setminus  \{j\}$ is replaced by a star branch of type $\mu$.
\end{lemma}
\begin{proof}
If one branch $i \in I$ receives a rank $c \le \mu$, then $c < \mu$ because $c \notin A_i$ and $\max A_i = \mu$.  Its visible set $\ctr(A_i, c) = \{c\} \cup \{x \in A_i : x > c\}$ therefore contains $\mu$. Because visible sets must be pairwise disjoint, at most one branch in $I$ can receive rank at most $\mu$.

Now consider a critical assignment. If some branch $i \in I \setminus \{j\}$ receives a rank $c_i \le \mu$, by the above observation, the minimum-weight branch $j$ must receive a rank $c_j > \mu$, contributing simply $\{c_j\}$. 
By the choice of $j$ and \cref{lem:dom}(b), $\wt(A_j) \le \wt(A_i) < \wt(\ctr(A_i, c_i))$, so \cref{lem:dom}(a) yields an admissible rank $c \notin A_j$ with $\ctr(A_j, c) \subseteq \ctr(A_i, c_i)$. 

We now reassign ranks: branch $j$ gets rank $c$ and branch $i$ gets rank $c_j$. 
The rank $c_j$ is admissible for branch $i$ because $c_j > \mu = \max A_i$, so branch $i$ now contributes $\{c_j\}$. 
The new contributions of $j$ and $i$ are contained in the old contributions of $i$ and $j$, respectively, while all other branch contributions remain unchanged. 
Thus, contributions remain pairwise disjoint and their union cannot increase in weight, so the new assignment is also critical. 
Under this assignment, each branch in $I \setminus \{j\}$ receives rank greater than $\mu$.
For the second claim, a star branch of type $\mu$ has admissible ranks $\{c : c>\mu\}$ and contributes $\{c\}$, exactly like a branch $i\in I$ that receives a rank above $\mu=\max A_i$. 

Hence, every valid assignment after the replacement is valid before it, with the same visible set, and conversely the critical assignment from (i) remains valid after the replacement. The two minima coincide, so $\Lab(A_1,\ldots,A_d)$ is unchanged since $\wt$ is injective.
\end{proof}

We combine the segment bound of \cref{lem:seg} with the star branch replacement rule of \cref{lem:star_rep} to bound the number of explicit branches and their space consumption. See \cref{fig:grouping} for an illustration.
\begin{lemma}[Size of critical sets]\label{lem:count}
Consider a vertex $w$ and a subset of its children whose subtrees contain $N_w$ vertices in total. Let $g_w$ be the number of distinct maximum ranks among their critical sets. Then, $g_w \le \sqrt{2N_w}$.

Moreover, consider distinct vertices $w_1, \ldots, w_k$, each having a subset of children such that their subtrees are mutually disjoint. Let $N$ be the total number of vertices in all these subtrees, and let $\mathcal A'$ be the collection of critical sets kept explicitly after applying \cref{lem:star_rep} across all these children. Then,
\begin{equation}\label{eq:groups}
 |\mathcal A'| \le \sqrt{2kN},
 \qquad
 \sum_{A \in \mathcal A'} \nint{A} \le \sqrt{|\mathcal A'|N} \le 2^{1/4} N^{3/4} k^{1/4}.
\end{equation}
\end{lemma}
\begin{proof}
A child whose critical set has maximum rank $\mu$ has at least $\mu$ edges ($\mu+1$ vertices). Subtrees with $g_w$ distinct maximum ranks $\mu_1<\cdots<\mu_{g_w}$ therefore have at least $1+2+\cdots+g_w$ vertices in total, so $g_w(g_w+1)/2 \le N_w$, which yields $g_w \le \sqrt{2N_w}$.

By \cref{lem:star_rep}, at most one critical set is kept per vertex $w_i$ \textit{and} distinct maximum, so vertex $w_i$ contributes at most $g_{w_i} \le \sqrt{2N_{w_i}}$ explicit sets, where $N_{w_i}$ is the total size of the subtrees of its considered children. Summing over all $k$ vertices and applying the Cauchy--Schwarz inequality,
\[
|\mathcal A'| \le \sum_{i=1}^k g_{w_i} \le \sum_{i=1}^k \sqrt{2N_{w_i}} \le \sqrt{2k \sum_{i=1}^k N_{w_i}} \le \sqrt{2kN}.
\]
The explicit sets belong to $|\mathcal A'|$ pairwise disjoint subtrees of total size at most $N$, so \cref{lem:seg} bounds their total number of segments by $\sqrt{|\mathcal A'|N}$, which is at most $(2kN)^{1/4} N^{1/2} = 2^{1/4} N^{3/4} k^{1/4}$.
\end{proof}

\begin{figure}[ht]
\centering
\begin{tikzpicture}[
  vtx/.style={circle,fill=black,inner sep=1.3pt},
  lab/.style={font=\scriptsize},
  fr/.style={draw=#1,fill=#1!15,thick},
]
\newcommand{\fringe}[4]{
  \node[vtx] at (#1,-1) {};
  \draw[fr=#2] (#1,-1) -- (#1-0.33,-1.75) -- (#1+0.33,-1.75) -- cycle;
  \node[lab,text=#2] at (#1,-2.05) {#3};
}
\begin{scope}
  \node[vtx,label={[lab]above:$v$}] (v) at (2.3,0) {};
  \fringe{0}{grpA}{$\{1,2,3,7\}\,$}{}
  \fringe{1.15}{grpA}{$\{1,5,7\}$}{}
  \fringe{2.3}{grpA}{$\{1,6,7\}$}{}
  \fringe{3.45}{grpB}{$\{1,4\}$}{}
  \fringe{4.6}{grpB}{$\{2,3,4\}$}{}
  \foreach \x in {0,1.15,2.3,3.45,4.6} \draw (v) -- (\x,-1);
  \node[font=\footnotesize] at (2.3,-2.6) {explicit visible sets};
\end{scope}
\draw[-{Stealth[length=3mm]},thick] (5.2,-0.9) -- node[above,lab,align=center]{compression}
  node[below,lab]{} (7.3,-0.9);
\begin{scope}[xshift=8.1cm]
  \node[vtx,label={[lab]above:$v$}] (w) at (1.6,0) {};
  \fringe{0}{grpA}{$\{1\raisebox{0.08em}{\scriptsize $-$}3,7\}\,\,$}{}
  \fringe{0.9}{grpB}{$\{1,4\}$}{}
  \draw (w) -- (0,-1); \draw (w) -- (0.9,-1);
  \node[draw=grpA,thick,rounded corners=2pt,lab,fill=grpA!8] (z7) at (2.25,-1.25) {$2\times[7]$};
  \node[draw=grpB,thick,rounded corners=2pt,lab,fill=grpB!8] (z4) at (3.35,-1.25) {$1\times[4]$};
  \draw (w) -- (z7.north); \draw (w) -- (z4.north);
  \node[font=\footnotesize] at (1.6,-2.6) {compressed visible sets};
\end{scope}
\end{tikzpicture}
\caption{Applying the $\mu$-star branch replacement rule (\cref{lem:star_rep}) and segment representation at $v$.}
\label{fig:grouping}
\end{figure}
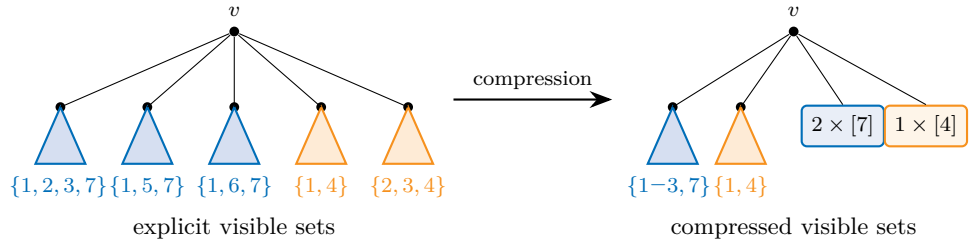

\subsection{Labeling with compressed sets}\label{sec:labeling}
Lam and Yue~\cite{LamY01} give a procedure \Labeling{} that computes a critical assignment of ranks to the branches at $v$ from the critical sets of the children. Running \Labeling{} bottom-up yields a critical ranking of the tree in linear total time. 

Our goal is to compute a critical assignment to the branches using compressed visible sets (cf.~\cref{sec:mpc-onepermax}). We show that \Labeling{} can be implemented (and thus a critical ranking produced) directly on these compressed visible sets without expanding individual branches. All star branches of a type except at most one survivor are labeled in a single resolution step; the survivor may take part in several.

\paragraph{Labeling procedure.} We first recall \Labeling{}. Its input consists of a vertex $v$ and the critical sets $A_1,\ldots,A_d$ of the child subtrees at $v$.  During \Labeling{}, every branch $i$ has a label $\ell_i\ge0$, where $\ell_i=0$ means unlabeled, and a \emph{remaining set} $\hat A_i=A_i\cap(\ell_i,\infty)$. A rank is \emph{free} if it lies in no remaining set and equals no label. Initially, all branches are unlabeled. The procedure repeats the following two steps until all labels are positive.
\begin{enumerate}[leftmargin=*]
\item \emph{Conflict.} Let $\alpha$ be the largest rank contained in at least
  two remaining sets, and let $\mathcal K$ be the set of branches whose
  remaining sets contain $\alpha$. If no rank is contained in two remaining
  sets, let $\alpha=0$ and let $\mathcal K$ be the set of unlabeled branches.
\item \emph{Resolution.} Let $\beta$ be the smallest rank above $\alpha$ in
  $\bigcup_i\hat A_i$, or $\beta=\infty$ if there is none, and let $\phi$ be
  the number of free ranks in $(\alpha,\beta)$. Then, we have one of the following cases:
  \begin{enumerate}[label=(\alph*)]
      \item If $\alpha>0$ and $\phi\ge|\mathcal K|-1$, choose $i_0\in\mathcal K$ whose remaining set restricted to $[\alpha]$ has minimum weight, and give branches in $\mathcal K\setminus\{i_0\}$ the $|\mathcal K|-1$ smallest free ranks above $\alpha$, in any order.
      \item If $\alpha=0$ and $\phi\ge|\mathcal K|$, give the branches in $\mathcal K$ the $|\mathcal K|$ smallest free ranks, in any order.
      \item Otherwise, let $x$ be the smallest free rank above $\beta$, and let $j$ be the branch whose remaining set contains the largest rank below $x$; this rank lies in $[\beta,x)$. Set $\ell_j=x$.
  \end{enumerate} 
\end{enumerate}
Note that increasing $\ell_i$ removes from $\hat A_i$ all ranks at most $\ell_i$. At termination, the visible set at $v$ is $\bigcup_i(\hat A_i\cup\{\ell_i\})$ and each branch receives its final label as its rank. Lam and Yue prove that the resulting ranking is critical~\cite[Section~4]{LamY01}.

\paragraph{Batched labeling.} In a standard execution of \Labeling{}, maintaining an explicit descriptor for each of the $z_i$ star branches of type $\mu_i$ would consume $\Omega(z_i)$ space, which can be as large as $\Omega(n)$. To achieve sublinear space, we run \Labeling{} with batched assignments. The $z_i$ star branches of type $\mu_i$ form a \textit{batch}: one of them (the potential survivor) is handled explicitly, and the other $z_i-1$ receive their ranks in a single step. The assigned ranks are added to a global set of batch label intervals $Z$, which are assigned in a final matching stage.
\begin{lemma}\label{lem:batched}
Let $v$ have $d$ branches with visible sets $A_1, \ldots, A_d$. Let $m$ branches be segment-represented using $M$ segments in total, and let the other branches be star branches of $g$ distinct types $\mu_1,\ldots,\mu_g$ with counts $z_1,\ldots,z_g$.

A critical assignment of ranks to the branches of $v$ and $\Lab(A_1,\ldots,A_d)$ can be computed using $O(M+m+g)$ total space and represented using $O(M+m+g)$ segments. Total runtime is polynomial.
\end{lemma}
\begin{proof}
We simulate the \Labeling{} procedure of \cite{LamY01} using our compressed representation. The proof proceeds in three parts: first, we identify the structural invariants of the labeling process (properties (P1)--(P3)); second, we show that at most one star branch per type (the survivor) undergoes relabeling or participates in subsequent conflicts; and third, we prove that batch labels can be stored collectively as segments to reduce space consumption.

\paragraph{Labeling properties.} Our proof relies on the following properties of the labeling process:
\begin{enumerate}[label=(P\arabic*),nosep,leftmargin=*]
    \item In resolution steps (a)--(b), if $(\alpha,\beta)$ contains enough free ranks, all conflicting branches except one \emph{survivor} receive the smallest free ranks above $\alpha$, in any bijection, in a single step; the survivor is a conflicting branch whose remaining set restricted to $[\alpha]$ has minimum weight, and for $\alpha=0$ there is no survivor.
    \item In resolution step (c), the branch owning the largest remaining rank below the smallest free rank $x$ above $\beta$ receives label $x$, removing at least one rank below $x$ from its remaining set.
    \item Labels only increase, remaining sets only lose their smallest ranks, a labeled branch with empty remaining set is never relabeled, and the final labeling is a critical assignment.
\end{enumerate}
Observe that (P1) and (P2) follow easily from the definition of \Labeling{}; (P3) follows from Lam and Yue's analysis~\cite[Section 4]{LamY01} combined with \cref{lem:star_rep}.

\subparagraph{At most one survivor per type.} Consider star branches of type $\mu$. While at least two star branches of type $\mu$ with $\mu>0$ are unlabeled, rank $\mu$ lies in two remaining sets, so the conflict satisfies $\alpha\ge \mu$. A (P2) step labels the owner of a remaining rank \textit{above} $\alpha$, which is never one of these star branches. Hence, branches of type $\mu$ are first labeled when a conflict $\alpha=\mu$ is resolved by a (P1) step, and all but at most one survivor must be labeled in that step since $\phi \ge |\mathcal K| - 1$. This means that we can designate one star branch per type as the only potential survivor, which the arbitrary tie-breaking in (P1) permits. After a branch is labeled above $\mu$, it has an empty remaining set and thus is never relabeled by (P3). Branches of type $0$ are labeled together in the final step with $\alpha=0$; for uniformity, we likewise designate one star branch of type $0$ to be tracked explicitly. We call these $g$ designated star branches (one per type) and the $m$ explicit branches \emph{ordinary} branches, whereas the remaining $z_i-1$ star branches of each type $\mu_i$ receive \emph{batch labels}.

\subparagraph{Compact label and remaining set representation.} We bound the space used throughout the execution. Consider first ordinary branches. The remaining sets associated with ordinary branches initially have at most $M+g$ segments. When the label of an ordinary branch changes to a rank $c'$, rank $c'$ is free and lies in no segment, so $c'$ causes the removal of an \textit{entire} segment from its remaining set; only a resolution step (b) with $\alpha=0$ may remove no segments, which occurs at most once per branch. This means that no segment is ever fragmented or occupies more than $O(1)$ words. Consequently, ordinary labels change at most $(M+g)+(m+g)$ times in total.

Consider now batch labels. Every non-ordinary star branch is labeled only once, in a (P1) step with $\alpha=\mu_i$ (which assigns $|\mathcal K|-1$ free ranks for $\mu_i>0$, or $|\mathcal K|$ for $\mu_i=0$). However, we cannot afford to store them explicitly. We store batch labels as a set of segments $Z$, \textit{without recording types}. When $\alpha=\mu_i$ is resolved (with $z_i>1$), we give the $z_i-1$ non-ordinary star branches the $z_i-1$ smallest free ranks assigned in this step, which (P1) permits (if the ordinary star branch of type $\mu_i$ is also labeled in this step, its rank is recorded in $\Lambda$).

These ranks span a block $[x,y]\subseteq(\alpha,\beta)$ that contains no remaining rank; its maximum $y$ is a batch label, and every rank in $[x,y]$ is a label after the step. We update $Z\leftarrow Z\cup([x,y]\setminus\Lambda)$, where $\Lambda$ is the set of labels of ordinary branches after the step. Thus, $Z$ is the union of at most $g$ contiguous blocks minus labels of ordinary branches, that is, each endpoint of a segment in $Z$ is an endpoint of a contiguous block $[x,y]$ or is adjacent to a label of an ordinary branch. Hence, $Z$ has $O(M+m+g)$ segments. Counts of unlabeled branches and their remaining sets also fit into $O(M+m+g)$ space.

Note that the operations of \Labeling{}, including locating $\alpha$ and $\beta$, counting and taking free ranks, comparing restricted sets, and deleting segments from remaining sets, are performed on segment lists incurring no asymptotic space overhead.

\subparagraph{Critical ranking and visible set output.} Let $\Lambda_{\mathrm{fin}}$ be the final labels of ordinary branches, so $|\Lambda_{\mathrm{fin}}|\le m+g$. At termination, every rank in a block $[x,y]$ is a label: otherwise some block contains a free rank $r<y$, and moving the batch label $y$ to $r$ keeps the assignment valid (as $r>\alpha=\mu$) while decreasing its weight, contradicting criticality (P3). Hence $Z$ is the union of at most $g$ blocks minus $\Lambda_{\mathrm{fin}}$ and consists of at most $2g+m$ segments. (During the execution, $Z$ may have $O(M+m+g)$ segments, since ordinary labels that later increase leave holes in blocks.)

We show how to match non-ordinary star branches to a batch label such that each star branch receives a label that exceeds its type. We have $g$ pairs $(\mu_i, z_i-1)$ and the segments of $Z$ such that $\sum_{i=1}^g (z_i - 1) = \sum_{[x,y]\in Z}(y-x+1)$. We sort both increasingly and match them in order, producing a valid ranking, which is critical by (P3). Note that if types $\mu<\mu'$ were matched to labels $c>c'$ by the execution, then $c'>\mu'>\mu$ and $c>c'>\mu'$, so the two labels can be exchanged, and repeating this sorts the matching. Since $Z$ has $O(m+g)$ segments and there are $O(g)$ type changes, the matching can be described using $O(m+g)$ space: compute prefix sums over the counts and note the labels where a type change occurs. 

The final visible set $\Lab(A_1,\ldots,A_d)$ has $O(M+m+g)$ segments: an explicit branch contributes at most $\nint{A_i}+1$ segments, and the star branches contribute $O(m+g)$ segments in total. Since \Labeling{} runs in polynomial time and so does each operation described, the total running time remains polynomial.
\end{proof}

\subsection{Solving a subtree on one machine}\label{sec:mpc-instance}
We use our size-threshold decomposition to divide the problem into smaller subproblems that can be solved locally. Consider a subtree $T_w$ with $N>B$ vertices. We use \cref{lem:size-threshold} with $\tau=\lfloor B\rfloor+1$ to retain each vertex $v$ with $|T_v|>B$. Consequently, attached subtrees are the maximal subtrees of size at most $B$ and our retained tree has at most $HN/B$ vertices.

We say that a subtree $T_u$ is \emph{solved} if its edges carry a critical
ranking of $T_u$. A solved subtree is never modified again. The following lemma relates the local memory size to the subtree size and the number of retained vertices, showing when $T_w$ can be solved in $O(1)$ rounds.

\begin{lemma}\label{lem:instance}
Let $S\ge n^{3/4}$. Let $T_w$ have $N>B$ vertices, $k$ of which are retained, and suppose every attached subtree of $T_w$ is solved. Let $N'$ be the number of vertices of $T_w$ outside its largest attached subtree. If $N'^{3/4}k^{1/4}=O(S)$, then $O(1)$ rounds suffice to compute a critical ranking of $T_w$ from its solved attached subtrees with $O(S)$ local memory and $O(N)$ total space.
\end{lemma}
\begin{proof}
We gather the information required to compute a critical ranking of the retained tree of $k$ vertices at a single machine. The first step is to compress the visible sets of its attached subtrees in parallel using the segment and star-based representations from \cref{sec:mpc-onepermax}. The computation of a critical ranking is then executed locally via the batched labeling approach of \cref{sec:labeling}.

\subparagraph{Attached subtree compression.}
For each attached subtree $T_u$, we have a record consisting of the edge $vu$ to its retained parent, its segment-represented critical set $\crit(T_u)$, and the maximum rank $\mu$ of $\crit(T_u)$ (with $\mu=0$ if $\crit(T_u)=\varnothing$). Each record has size $O(\sqrt N)$ and their total size is $O(N)$ by~\cref{lem:seg}. Let $T_{u^*}$ be the largest attached subtree. Its critical set has $O(\sqrt{n})=O(S)$ segments, so we keep it explicitly and apply the bounds below to the remaining attached subtrees, which contain at most $N'$ vertices.

Leaf attached subtrees ($\mu=0$) at each retained vertex $v$ are stored by their count as star branches of type $0$, while for $\mu\ge 1$, by \cref{lem:star_rep} applied at $v$ among the roots of its remaining attached subtrees, we only need to keep explicitly a minimum-weight critical set in each $(v,\mu)$-group and replace the rest by star branches of type $\mu$. Using the grouped-selection primitive from \cref{lem:select}, the minimum (for $\mu\ge 1$) and size of each $(v,\mu)$-group can be computed in $O(1)$ rounds. Let $G$ be the number of $(v,\mu)$-groups across the remaining attached subtrees. Since all $k$ retained vertices lie outside $T_{u^*}$, we have $k\le N'$. There are at most $k$ groups with $\mu=0$ (contributing no segments), and by \cref{lem:count} applied to the $k$ retained vertices---each with its non-leaf attached child subtrees outside $T_{u^*}$, which are pairwise disjoint and contain at most $N'$ vertices---there are at most $\sqrt{2kN'}$ groups with $\mu\ge 1$, whose kept critical sets have at most $2^{1/4}N'^{3/4}k^{1/4}$ segments in total. Thus $G\le \sqrt{2kN'}+k = O(N'^{3/4}k^{1/4})$, and the \emph{compressed instance}, consisting of these $G$ group records together with $\crit(T_{u^*})$ and the parent pointers of the $k$ retained vertices, has total size $O(N'^{3/4}k^{1/4}+\sqrt{n})=O(S)$ and can be collected onto a single machine.

\subparagraph{Ranking the retained tree.}
We process retained vertices in postorder. For each retained vertex $v$, we compute $\crit(T_v)$ using \cref{lem:batched}. Its children are either roots of attached subtrees or retained vertices, whose critical sets have already been computed, and we can apply induction over the postorder proving that the ranking of $T_v$ is indeed critical by \cref{lem:dom}(c) and \cref{lem:star_rep}(ii).

It remains to show that the local memory is never exceeded throughout the computation. Since rankings of attached subtrees can be compressed into $O(S)$ space and computing the critical ranking at a retained vertex can be done without increasing space usage by \cref{lem:batched}, the sole memory violation may arise from solved retained vertices. Since the subtrees of retained vertices waiting for their parent are pairwise disjoint and at most one of them contains $T_{u^*}$, their critical sets have $O(\sqrt{n}+\sqrt{kN'})=O(S)$ segments in total by~\cref{eq:cs}, using $k\le N'$; they are discarded once the parent is processed.

\subparagraph{Mapping labels back to star branches.}
At each retained vertex $v$, the assignment computed by \cref{lem:batched} gives, for every type $\mu$, the batch labels of the star branches of type $\mu$ at $v$ as a sorted list of segments; over all retained vertices these lists have $O(G+k)=O(S)$ segments. We sort the roots of the corresponding attached subtrees arbitrarily, index them $1,2,\ldots$, and give the root with index $i$ the $i$-th smallest label of its type. Sending segment endpoints and prefix counts to the machines storing the sorted roots via a broadcast tree takes $O(1)$ rounds.
\end{proof}

\subsection{MPC Algorithm}\label{sec:mpc-stages}
We prove our main MPC result (\cref{thm:mpc}). Our algorithm processes the tree bottom-up in stages governed by increasing size thresholds $S=B_1<B_2<\cdots<B_L\ge n$, where the first stage solves maximal subtrees of size at most $B_1$ locally and each subsequent stage $j$ computes a critical ranking for every maximal subtree $T_w$ of size at most $B_j$ from the already-solved subtrees. An initial $O(\log D)$-round preprocessing step roots the tree and computes subtree sizes, depths, and preorder numbers so that vertices can identify their stage-$j$ subtrees in $O(1)$ rounds.

We distinguish two cases by the height $H$ (where $H-1\le D\le 2(H-1)$), which every machine knows after preprocessing. Recall that $S\ge n^{3/4+\delta}$. If $H\le n^{2\delta}$ (\cref{sec:mpc-a}), we choose the thresholds $B_j$ so that every stage-$j$ subtree satisfies the memory condition of \cref{lem:instance} and is solved in $O(1)$ rounds; the thresholds reach $n$ in $L=O(\log(1/\delta))=O(1)$ stages. If $H>n^{2\delta}$ (\cref{sec:mpc-b}), we instead double the thresholds at each stage. This ensures that the retained vertices of each stage-$j$ subtree form a simple path, which we partition into subpaths and solve bottom-up in $O(1)$ rounds per subpath via \cref{lem:instance}, yielding a geometric series of rounds that sums to $O\bigl(\frac{n^{3/4}D^{1/4}}{S}+\log D\bigr)$.

\subsubsection{Small diameter case}\label{sec:mpc-a}
A threshold parameter $B_j$ determines the size of the maximal subtrees processed in stage $j$, for $L$ stages:
\begin{equation}\label{eq:thresholds}
 B_1=S,\qquad B_{j+1}=S\Bigl(\frac{B_j}{H}\Bigr)^{1/4}, \qquad j \in [L-1].
\end{equation}
Note that $B_j = H (S/H)^{\frac{4}{3}(1-4^{-j})}$ is bounded by $S^{4/3}$ and that $L$ will be chosen to guarantee $B_L \ge n$, so that the entire tree is processed. In particular, since $H\le n^{2\delta}$, we have $S\ge n^{3/4+\delta}\ge n^{3/4+\delta/2}H^{1/4}$, so $B_L\ge n$ holds for  $L=O(\log(1/\delta))=O(1)$. Our small-diameter algorithm is presented next.

\subparagraph{Preprocessing and first stage.}
The first step is to root the tree, compute subtree sizes, depths, and preorder numbers via \cref{lem:preorder} in $O(\log D)$ rounds. Every vertex also learns the subtree size of its parent in one additional round and $H$ via global communication. All machines can then determine the thresholds and $L$.

The first stage processes maximal subtrees $T_w$ of size at most $B_1=S$ (where $w$ is the tree root or $|T_{\mathrm{parent}(w)}|>S$). Using preorder numbers and prefix sums over $|T_w|$, vertices with $|T_v|\le S$ identify their maximal-subtree root $w$ and pack these disjoint subtrees onto machines of $O(S)$ vertices each in $O(1)$ rounds. Each machine then computes a critical ranking for each of its assigned subtrees via the sequential algorithm of Lam and Yue~\cite{LamY01}, which uses linear space, and stores the resulting visible set $\crit(T_w)$ in segment representation.

\subparagraph{Stage $j\ge2$.}
We compute a critical ranking for each maximal subtree $T_w$ with $B_{j-1}<|T_w|\le B_j$. Each root $w$ with $B_{j-1}<|T_w|\le B_j$ recognizes itself locally since it is the tree root or $|T_{\mathrm{parent}(w)}|>B_j$. All vertices $v$ that satisfy $|T_v|\le B_j$ and $|T_{\mathrm{parent}(v)}|>B_{j-1}$ are part of this stage: $v$ is retained if $|T_v| > B_{j-1}$, otherwise it roots an attached subtree solved in a previous stage. Each vertex $v$ learns the root $w$ of its subtree (the highest ancestor with $|T_w|\le B_j$) using preorder numbers via standard $O(1)$-round MPC primitives.

We can then apply \cref{lem:instance} to solve each $T_w$-subtree. Each subtree has at most $B_j$ vertices with at most $HB_j/B_{j-1}$ retained vertices by \cref{lem:size-threshold}, so
\[
 B_j^{3/4}\left(\frac{HB_j}{B_{j-1}}\right)^{1/4} = B_j\Bigl(\frac{H}{B_{j-1}}\Bigr)^{1/4}=S
\]
by~\cref{eq:thresholds}, and \cref{lem:instance} solves each subtree in $O(1)$ rounds using $O(S)$ local memory and $O(|T_w|)$ space per subtree, ensuring linear total space since subtrees are disjoint. 

Since $B_L\ge n$, the whole tree is solved after stage $L$. By induction, assuming that each subtree processed in earlier stages is critically ranked, the final ranking is optimal by \cref{lem:instance} and the fact that a critical assignment to the branches leads to a critical ranking of the whole tree (\cref{lem:dom}(c)). Each stage takes $O(1)$ rounds, so the total round complexity is $O(\log D+L)=O(\log D)$, with $O(S)$ local memory and $O(n)$ total space.
    
\subsubsection{Large diameter case}\label{sec:mpc-b}
We now modify our algorithm from \cref{sec:mpc-a} to handle the case $H>n^{2\delta}$, obtaining a trade-off between round complexity and local memory. We use doubling thresholds
\[
 B_j = \min\{2^{j-1}S, \, n\} \qquad (j = 1, 2, \ldots, J), \qquad J = \lceil \log_2(n/S) \rceil + 1.
\]
In addition to the preprocessing of \cref{sec:mpc-a}, every vertex learns the largest subtree size among its children in $O(1)$ rounds, and whenever a vertex learns the root $w$ of its subtree, it also learns $|T_w|$ and $\mathrm{depth}(w)$.

Stage $1$ is as in \cref{sec:mpc-a}. Stage $j\ge2$ processes each maximal subtree $T_w$ with $B_{j-1} < |T_w| \le B_j$. Recall that a vertex $v \in T_w$ is retained if $|T_v| > B_{j-1}$. If a retained vertex had two retained children $a$ and $b$, then $|T_a| + |T_b| > 2B_{j-1} \ge B_j \ge |T_w|$, a contradiction. Thus, the retained vertices of $T_w$ form a simple path $w=h_1,\ldots,h_p$ with $p\le H$, where $h_{i+1}$ is the unique retained child of $h_i$ ($1 \le i < p$) and every child of $h_p$ roots an attached subtree. Each retained vertex $h_i$ knows its path index $i = \mathrm{depth}(h_i) - \mathrm{depth}(w) + 1$ and whether it is $h_p$, i.e., whether all its children have subtree size at most $B_{j-1}$.

Since the entire subtree $T_w$ may not satisfy the memory condition of \cref{lem:instance}, we partition the retained path $h_1, \ldots, h_p$ into contiguous subpaths and process them bottom-up in sub-stages. Specifically, using step sizes
\[
 \Delta N = \Bigl\lceil S\Bigl(\frac{B_j}{H}\Bigr)^{1/4}\Bigr\rceil, \qquad\Delta k = \Bigl\lceil S\Bigl(\frac{H}{B_j}\Bigr)^{3/4}\Bigr\rceil,
\]
each retained vertex $h_i$ locally computes its sub-stage index
\[
 m(i) = 1 + \Bigl\lfloor \frac{|T_w| - |T_{h_i}|}{\Delta N} \Bigr\rfloor + \Bigl\lfloor \frac{i - 1}{\Delta k} \Bigr\rfloor.
\]
Since $|T_w| - |T_{h_i}| < B_j$ and $i - 1 < p \le H$ are increasing in $i$, the index $m(i)$ partitions $h_1, \ldots, h_p$ into at most
\[
 r_j \le 1 + \frac{B_j}{\Delta N} + \frac{H}{\Delta k} \le 1 + \frac{2 B_j^{3/4} H^{1/4}}{S}
\]
subpaths $P_m = (h_a, \ldots, h_b)$, each satisfying $b - a <\Delta k$ and $|T_{h_a} \setminus T_{h_b}| = |T_{h_a}| - |T_{h_b}| < \Delta N$.
We process the subpaths $P_m = (h_a, \ldots, h_b)$ in reverse order $m = r_j, r_j - 1, \ldots, 1$ (skipping empty $P_m$), so that $\crit(T_{h_{b+1}})$ is already known when $P_m$ is reached:

\begin{enumerate}
 \item First, we apply \cref{lem:instance} to $T_{h_b}$ with $B=|T_{h_b}|-1$, so that $h_b$ is its only retained vertex. Its attached subtrees ($T_{h_{b+1}}$ if $b<p$, and the subtrees attached to $h_b$) are solved, and $N'^{3/4}k^{1/4}\le n^{3/4}\le S$ with $k=1$, so this step computes $\crit(T_{h_b})$ in $O(1)$ rounds.
 \item Second, if $a<b$, we apply \cref{lem:instance} to $T_{h_a}$ with $B=|T_{h_b}|$. Its retained vertices are $h_a,\ldots,h_{b-1}$, its largest attached subtree is the solved $T_{h_b}$, and $N'=|T_{h_a}\setminus T_{h_b}|<\Delta N$. Since
 \[
  N'^{3/4}(b-a)^{1/4}<(\Delta N)^{3/4}(\Delta k)^{1/4}=O(S),
 \]
 this step computes $\crit(T_{h_a})$ in $O(1)$ rounds using $O(S)$ local memory.
\end{enumerate}

Processing all $r_j$ sub-stages of stage $j$ across all disjoint stage-$j$ subtrees $T_w$ in parallel takes $O(r_j)$ rounds and $O(n)$ total memory. Summing over stages $j = 2, \ldots, J$ and using $B_j\le 2^{j-J+1}n$, the geometric series with ratio $2^{3/4} > 1$ gives a total round complexity of
\[
 O(\log D) + \sum_{j=2}^{J} O\Bigl(1 + \frac{B_j^{3/4} H^{1/4}}{S}\Bigr)
 = O(\log D + J) + O\Bigl(\frac{H^{1/4}}{S}\sum_{j=2}^{J} (2^{j-J+1}n)^{3/4}\Bigr)
 = O\Bigl(\frac{n^{3/4} D^{1/4}}{S} + \log D\Bigr),
\]
where the last equality uses $H\le D+1$ and $J=O(\log(n/S))=O(\log n)=O(\log H)=O(\log D)$, since $H>n^{2\delta}$ and $\delta$ is fixed. Together with \cref{sec:mpc-a}, this completes the proof of \cref{thm:mpc}.

\section*{AI Usage}
\cref{thm:mpc} was developed through interaction with ChatGPT 6 Pro. We prompted the model with our unpresented idea for $O(\log D)$-round 2-approximate edge ranking in MPC. We then considered a strongly sublinear representation of critical sets~\cite{LamY01} that led to the presented MPC algorithm. While trying to optimize the MPC algorithm for the general case, the authors tested various hypotheses about its hardness in PRAM, and among them key algorithmic ideas leading to \cref{thm:pcomplete} and \cref{thm:approx-nc-main} emerged. We have rewritten the paper for human-centered exposition. We used additional AIs for review and copyediting. The authors take full responsibility for the manuscript’s content and the correctness of results.

\appendix
\section{Further related work}\label{sec:relwork}
\paragraph{Node ranking, tree-depth, and generalized rankings.}
A closely related problem is \emph{node ranking} (also called vertex ranking or ordered coloring~\cite{KatchalskiMS95}), which assigns positive integer ranks to vertices so that every path between two vertices of equal rank contains a vertex of larger rank. On trees, sequential node ranking was solved in $O(n\log n)$ time by Iyer, Ratliff, and Vijayan~\cite{IyerRV88} and optimally in $O(n)$ time by Sch{\"a}ffer~\cite{Schaffer89}. On general graphs, the minimum node ranking number coincides with \emph{tree-depth}~\cite{BodlaenderDJKKMT98,NesetrilO06}, a central parameter in structural graph theory and distributed and parameterized algorithms~\cite{NesetrilO16distributed,BorowieckiDO23,FominFMRT24}. Despite this close connection—indeed, an edge ranking of $T$ is equivalently a node ranking of its line graph $L(T)$—node ranking on trees does not share the parallel hardness of edge ranking: parallel algorithms based on tree contraction~\cite{MillerR85} placed optimal tree node ranking in $\NC$~\cite{LiangDL90,delaTorreG91,delaTorreGP92}, culminating in an optimal $O(\log n)$-time, $O(n)$-work PRAM algorithm~\cite{KashemR04}. Structurally, assigning a single rank to a vertex separates all of its subtrees at once, whereas the edges incident to a high-degree vertex form a clique in $L(T)$ and compete for distinct visible ranks. When the maximum degree is bounded, this clique bottleneck disappears, which explains why exact edge ranking (and generalized $c$-edge-ranking~\cite{ZhouKN98}) admits $\NC$ algorithms on bounded-degree trees~\cite{delaTorreGS95} and bounded-degree partial $k$-trees~\cite{KashemZN98} (see also~\cite{KashemZN00} for generalized vertex rankings of partial $k$-trees, and~\cite{Yue96,LamY98hard} for NP-hardness on general graphs).

\paragraph{Weighted tree search and optimization variants.}
Through its equivalence to edge-query binary search on trees~\cite{linial1985searching,BFN99,OnakP06,MozesOW08,D08}, tree edge ranking has also been extensively studied in weighted settings where edges or vertices have non-uniform query costs~\cite{CicaleseJLV12,DereniowskiKUZ17,DereniowskiW24}. Unlike unweighted tree edge ranking, many weighted variants become strongly NP-hard even sequentially and are studied through approximation algorithms~\cite{CicaleseJLV12,DereniowskiKUZ17}. Related optimization problems include finding minimum-height join trees in relational query processing~\cite{IbarakiK84,DereniowskiK06} and computing spanning trees of general graphs with minimum edge ranking~\cite{MakinoUI01}.

\paragraph{$\Pc$-completeness and greedy tree procedures.}
Within $\Pc$-completeness theory~\cite{Ladner75,Goldschlager77,GoldschlagerSS82,Cook85,KarloffR89,GreenlawHR95}, showing that a specific greedy heuristic or subroutine is $\Pc$-complete does not by itself rule out fast parallel algorithms for the underlying optimization problem, as alternative non-greedy $\NC$ strategies sometimes exist (cf.\ tree precedence scheduling~\cite{HelmboldM87} or tree node ranking~\cite{delaTorreGP92}). In particular, de~la~Torre, Greenlaw, and Sch{\"a}ffer~\cite{delaTorreGS95} proved that their \textsc{FillSlots} subroutine is $\Pc$-complete, while leaving open whether tree edge ranking itself lies in $\NC$~\cite[Problem~B.1.2]{GreenlawHR95}. Our reduction settles this question directly for the decision and constant-error approximation problems on trees of constant diameter.

\paragraph{Massively parallel computation.}
The MPC model~\cite{KarloffSV10,Goodrich99,GoodrichSZ11,BeameKS17,AndoniNOY14} has been studied extensively over the past two decades (see~\cite{im2023massively} for a survey), with major progress on connectivity, shortest paths, and minimum spanning trees~\cite{AndoniNOY14,AndoniSSWZ18,BehnezhadDELM19,ACP20,FGG22,CoyC23,BalliuLMOU23,CoyCMM25,AzarmehrBJLMZ25}, local symmetry breaking~\cite{CLM+18,ghaffari2018improved,ABB+19,behnezhad2019exponentially,ghaffari2019sparsifying,CFG+19,FGG23,GP24}, string dynamic programming~\cite{ImMS17,HSS19,MEG+21}, and clustering~\cite{LackiMOS20,EpastoMMZ22}, alongside conditional round lower bounds in the strongly sublinear memory regime~\cite{RoughgardenVW18,GhaffariKU19,CDP21,NanongkaiS22}. On trees specifically, $O(\log D)$-round deterministic MPC algorithms are known for rooting, connectivity, and dynamic programming~\cite{BBD+18,BalliuLMOU23,GuptaLMPSSSUV23}. Notably, however, general MPC tree dynamic programming frameworks~\cite{BBD+18,GuptaLMPSSSUV23} require each subtree's DP state to fit in constant or polylogarithmic space so that paths and branches can be contracted locally. Because exact tree edge ranking requires passing critical sets of up to $\Omega(n)$ ranks—and combining them at high-degree vertices is $\Pc$-complete—such frameworks do not apply directly.

\begingroup
\small
\bibliographystyle{alphaurl}
\bibliography{refs}
\endgroup

\end{document}